\documentclass[twocolumn]{autart}    

\usepackage{amssymb,latexsym,amsfonts,amsmath}
\usepackage{wrapfig} 

\newenvironment{authorbio}[2][]{%
	\par\addvspace{1em}%
	\noindent
	\if\relax\detokenize{#1}\relax 
	\else
	\begin{wrapfigure}{l}{0.8in} 
		\vspace{-5pt} 
		\includegraphics[width=1in,height=1.25in,clip,keepaspectratio]{#1}
	\end{wrapfigure}%
	\fi
	\noindent\textbf{#2} 
}{\par\addvspace{2em}}

\usepackage{mathrsfs}

\usepackage{algorithm}
\usepackage{algcompatible}

\usepackage{natbib}

\usepackage{xspace}

\usepackage[nocomma]{optidef}

\usepackage{graphicx}
\usepackage{paralist}
\usepackage{dsfont}
\usepackage{caption}
\usepackage{booktabs}
\usepackage{subfig}
\usepackage{eso-pic}
\usepackage{epsfig}
\usepackage{mathtools}
\usepackage{xfrac}
\usepackage {url}
\usepackage{epstopdf}
\usepackage{tikz}
\usepackage{multirow}
\newcommand*{\myalign}[2]{\multicolumn{1}{#1}{#2}}
\usepackage{booktabs}       
\usepackage[utf8]{inputenc} 
\usepackage[T1]{fontenc}    
\usepackage{nicematrix}
\usepackage{IEEEtrantools}
\usetikzlibrary{arrows,calc,decorations.markings,math,arrows.meta}
\newtheorem{theorem}{Theorem}[section]
\newtheorem{lemma}[theorem]{Lemma}
\newtheorem{problem}[theorem]{Problem}

\newtheorem{definition}[theorem]{Definition}

\newtheorem{remark}[theorem]{Remark}
\newtheorem{assumption}[theorem]{Assumption}
\numberwithin{equation}{section}

\usepackage[many]{tcolorbox}
\usetikzlibrary{calc,decorations.pathreplacing, calligraphy}
\tcbuselibrary{skins}
\tikzset{mybrace/.style={thick, decorate, decoration={calligraphic brace, amplitude=1.5mm, raise=-1mm}}}

\newtcolorbox{resp}[1][]{%
	enhanced jigsaw,%
	colback=gray!5!white,%
	colframe=gray!80!black,%
	size=small,%
	boxrule=1pt,%
	halign title=flush center,%
	coltitle=black,%
	breakable,%
	drop shadow=black!50!white,%
	attach boxed title to top left={xshift=1cm,yshift=-\tcboxedtitleheight/2,yshifttext=-\tcboxedtitleheight/2},%
	minipage boxed title=3cm,%
	boxed title style={%
		colback=white,%
		size=fbox,%
		boxrule=1pt,%
		boxsep=2pt,%
		underlay={%
			\coordinate (dotA) at ($(interior.west) + (-0.5pt,0)$);
			\coordinate (dotB) at ($(interior.east) + (0.5pt,0)$);
			\begin{scope}[gray!80!black]
				\fill (dotA) circle (2pt);
				\fill (dotB) circle (2pt);
			\end{scope}
		}%
	},%
	#1%
}

\definecolor{mydarkred}{HTML}{990000}

\newcommand{\R}{{\mathbb{R}}}

\newcommand{\N}{{\mathbb{N}}}

\newcommand{\ie}{{\it i.e.}}
\newcommand{\eg}{{\it e.g.}}

\newcommand{\Let}{:=}

\DeclareFontFamily{U}{stix2bb}{}
\DeclareFontShape{U}{stix2bb}{m}{n} {<-> stix2-mathbb}{}

\NewDocumentCommand{\stixbbdigit}{m}{%
	\text{\usefont{U}{stix2bb}{m}{n}#1}%
}
\newcommand{\bbzero}{\stixbbdigit{0}}
\newcommand{\Vp}{{\dot{\mathcal V}(x)}}
\newcommand{\Vps}{{\dot{\mathcal V}_i(x_i)}}
\newcommand{\V}{{\mathcal V(x)}}
\newcommand{\Vs}{{\mathcal V_i(x_i)}}
\newcommand{\VecF}{{\mathcal F_i}}
\newcommand{\Sig}{\Sigma_i\!=\!(A_i,B_i, D_i, \mathcal F_i, \mathcal G_i)}
\newcommand{\Siginf}{\Sigma = \mathcal I(\Sigma_i)_{i\in\N^+}}
\newcommand{\SiginfK}{\Sigma_K = \mathcal I_K(\Sigma_i)_{i\in\N^+}}
\newcommand{\Siginfset}{\Sigma\!=\!(f,g, X,U)}
\newcommand{\SiginfsetK}{\Sigma_K\!=\!(f,g,K, X,U)}
\newcommand{\conblock}{\begin{bmatrix} \Psi_i(x_i)P_i^{-1}\\ \mathcal G_i(x_i)\mathbf K_i(x_i)\\ \bbzero_{\sigma_i\times n_i}\end{bmatrix}}
\newcommand{\conblockD}{\begin{bmatrix} \Psi_i(x_i)P_i^{-1}\\ \mathcal G_i(x_i)\mathbf K_i(x_i)\end{bmatrix}}

\newcommand{\Diag}[2]{%
	\underset{ #1\in\N^+}{\operatorname{diag}}(#2)%
}
\newcommand{\NPS}{\textsc{ctia-NPS}}
\newcommand{\MATLAB}{\textsc{Matlab}\xspace}
\usepackage{dsfont}

\makeatletter
\long\def\@maketablecaption#1#2{%
	\@tablecaptionsize
	\global\@minipagefalse
	\hbox to \hsize{%
		\parbox[t]{\hsize}{#1. #2}%
	}%
}
\makeatother

\makeatother
\allowdisplaybreaks
\begin{document}

\begin{frontmatter}
	
	\title{Data-Driven Global Stabilization of Unknown Infinite Networks\thanksref{footnoteinfo}}

\thanks[footnoteinfo]{A. Mironchenko has been supported by the Heisenberg grant (MI 1886/3-1) of the German Research Foundation (DFG). Corresponding author:  Mahdieh Zaker (m.zaker2@newcastle.ac.uk).}

\author[UK1]{Mahdieh Zaker}\ead{m.zaker2@newcastle.ac.uk},
\author[DU]{Andrii Mironchenko}\ead{andrii.mironchenko@uni-bayreuth.de}, \author[UK1]{Amy Nejati}\ead{amy.nejati@newcastle.ac.uk},  \author[UK1]{Abolfazl Lavaei}\ead{abolfazl.lavaei@newcastle.ac.uk} 

\address[UK1]{School of Computing, Newcastle University, United Kingdom}
\address[DU]{Department of Mathematics, University of Bayreuth, Germany}                               


\begin{abstract}   
	This paper develops a direct data-driven framework for infinite networks with unknown nonlinear polynomial subsystems, enabling the synthesis of controllers that ensure {the entire network is uniformly globally asymptotically stable (UGAS)}.
	To address scalability challenges arising from high dimensionality, we develop a data-driven approach to construct an input-to-state stable (ISS) Lyapunov function and its corresponding controller for each unknown subsystem using only a single set of \emph{noise-corrupted} input–state trajectories collected from that subsystem. Once each subsystem admits a data-driven ISS Lyapunov function, we leverage a compositional small-gain framework for \emph{infinite-dimensional} spaces to construct a global control Lyapunov function and its associated controller, thereby ensuring UGAS of the entire infinite network. The effectiveness of the proposed data-driven approach is demonstrated through three case studies, including infinite networks of spacecraft, Lorenz chaotic systems, and an academic example with a state-dependent control input matrix.
\end{abstract}

\end{frontmatter}

\section{Introduction}

We live in a world of networks, where transportation systems, electrical grids, social networks, and engineering systems are becoming increasingly interconnected and complex each year.
Establishing stability and/or safety guarantees for such large-scale networks remains a significant challenge, largely due to the limited scalability of existing analytical and control criteria.
In many practical settings, including connected vehicles, smart cities, drone formations, and traffic networks, subsystems (also referred to as agents) may dynamically join or leave the network.
As a result, the network size is often {unknown} and may {vary} over time~\citep{bamieh2002distributed, jovanovic2005ill,noroozi2021set,noroozi2022relaxed}. Modeling such systems as finite networks with fixed size therefore yields oversimplified descriptions that fail to capture the inherent complexity of practical scenarios.

To address this challenge, the concept of \emph{infinite} networks has been introduced, in which finite large-scale networks with possibly unknown size are overapproximated by \emph{countably infinite} networks of finite-dimensional subsystems~\citep{dashkovskiy2020stability}. 
A representative example is a traffic network composed of vehicles on a roadway: not only is it impractical to determine the total number of vehicles, \emph{i.e.}, no predefined upper bound exists, but the system typically involves a large and dynamically varying population of agents.

The behavior and performance of an interconnected network, whether finite or infinite, are governed by the interactions among its subsystems and their combined impact on the overall dynamics. As a result, disturbances or failures stemming from a single subsystem can propagate through the interconnections, potentially leading to cascading effects that significantly compromise overall network performance. A foundational role for systematic analysis of such propagation phenomena in interconnected systems ~\citep{nesic2002integral,Son08,isidori2013nonlinear} is played by the \emph{input-to-state stability (ISS)} framework, initially introduced by~\citet{sontag1989smooth}. 
ISS has unified internal stability with robustness to external inputs and has fundamentally shaped nonlinear systems theory by enabling systematic methods for nonlinear stabilization, robust observer design, and the stability analysis of interconnected nonlinear control systems~\citep{Mironchenko2023ISS,karafyllis2019input,mironchenko2020input,Son08}.

The ISS of subsystems is ensured if there exist so-called ISS Lyapunov functions together with corresponding Lyapunov gain functions. These gains encode the interconnection topology and quantify the influence of neighboring subsystems on each other. Having this information about the interconnection structure, the stability of the interconnection can be analyzed by employing classical ISS small-gain theory that has been developed for finite networks with nonlinear ODE components~\citep{jiang1994small, jiang1996lyapunov, dashkovskiy2007iss, dashkovskiy2010small}. 
Namely, if gains satisfy a suitable small-gain condition, stability of the overall network is guaranteed~\citep{jiang1994small, jiang1996lyapunov, dashkovskiy2007iss, dashkovskiy2010small}.

\subsection{Key Challenges and Recent Studies}

Considering the importance of ISS, significant progress has been made in extending the ISS framework to \emph{richer classes} of finite-dimensional systems, including switched, hybrid, and impulsive dynamics~\citep{mancilla2001converse,cai2005results,dashkovskiy2013input}. More recently, ISS theory has been generalized to infinite-dimensional systems, encompassing time-delay systems~\citep{pepe2006lyapunov,chaillet2023iss}, partial differential equations~\citep{karafyllis2019input}, discrete-time systems~\citep{bachmann2022nonlinear,dashkovskiy2023some}, and general evolution equations in Banach spaces~\citep{schwenninger2020input,mironchenko2020input}. A comprehensive overview of these developments is provided by~\citet{mironchenko2020input}. 

Stability criteria developed for finite networks of ISS systems cannot be transferred to infinite networks in a straightforward manner. For instance, infinite cascade interconnections of ISS systems are not necessarily ISS, in contrast to finite networks. Building on significant advances in infinite-dimensional ISS theory~\citep{mancilla2001converse,cai2005results,dashkovskiy2013input,pepe2006lyapunov,chaillet2023iss,karafyllis2019input,bachmann2022nonlinear,dashkovskiy2023some,schwenninger2020input,mironchenko2020input}, together with well-established nonlinear small-gain criteria for stability analysis of finite networks of nonlinear systems~\citep{dashkovskiy2007iss,dashkovskiy2010small}, recent works have laid the theoretical foundations for stability analysis of infinite networks. Motivated by this interplay, several studies have developed small-gain theorems specifically tailored to infinite networks; see, \emph{e.g.,} the works by~\citet{kawan2020lyapunov,dashkovskiy2019stability,dashkovskiy2020stability,mironchenko2021nonlinear,KMZ23}.

Beyond the intrinsic challenges posed by infinite networks, many real-world applications are further complicated by the absence of accurate mathematical models for the underlying subsystems. As a result, {despite the significance of the aforementioned results, their applicability remains limited by the strong assumption that the system dynamics are known}. To address the challenges posed by unknown system models, data-driven control strategies have emerged and can be broadly classified into two categories: (i) indirect and (ii) direct approaches. \emph{Indirect} data-driven methods seek to reconstruct the unknown dynamics via system identification techniques; however, achieving an accurate model can be computationally intensive, particularly for high-dimensional complex nonlinear systems~\citep{formentin2014comparison}. Moreover, even when a reliable model is obtained, the subsequent design of controllers that guarantee input-to-state stability typically relies on model-based techniques. Consequently, such approaches involve a \emph{twofold} source of complexity: model identification followed by controller synthesis using conventional model-based methods. In contrast, \emph{direct} data-driven methods aim to bypass system identification altogether and instead provide formal stability analysis directly from observed data~\citep{martin2023guarantees}.

\subsection{Literature Overview on Data-Driven Control}

A growing body of literature has investigated data-driven frameworks for establishing stability certificates and synthesizing controllers for dynamical systems with unknown models. In this direction, pioneering works by~\citet{depersis2020tac} and~\citet{guo2021data} propose data-driven methodologies for stability analysis of persistently excited linear time-invariant systems based on Willems et al.'s fundamental lemma~\citep{willems2005note}, and nonlinear polynomial systems, respectively. Robust data-driven model predictive control for linear time-invariant systems is developed by~\citet{berberich2020data}, while \citet{taylor2021towards} present a data-driven framework for robust control synthesis of nonlinear systems under model uncertainty. Data-driven stability analysis for unknown systems has also been addressed in several settings, including switched systems~\citep{kenanian2019data}, continuous-time systems~\citep{boffi2021learning}, and discrete-time systems~\citep{lavaei2022data}. In the work by~\citet{zhou2022neural}, a framework is proposed for stabilizing unknown nonlinear systems by jointly learning a neural Lyapunov function and a nonlinear controller, with formal stability guarantees enforced via Satisfiability Modulo Theories (SMT) solvers. 

By employing overapproximation techniques to characterize the set of polynomial dynamics consistent with noisy data, \citet{chen2025data} construct an ISS Lyapunov function and a corresponding ISS controller for unknown nonlinear input-affine systems with polynomial dynamics. More recently, a data-driven feedback linearization approach has been explored by~\citet{depersis2025feedback}, where the change of coordinates and the linearizing controller are learned from a library of candidate functions using a single noise-free trajectory. In addition, \citet{zaker2025certified} investigate a data-driven approach for synthesizing \emph{incremental} ISS controllers for unknown nonlinear polynomial systems. Despite addressing stability analysis and controller synthesis using data, these approaches are all inherently tailored to \emph{monolithic}, low-dimensional systems. Consequently, they do not straightforwardly scale to large-scale interconnected networks, let alone infinite ones.

A limited number of studies have developed data-driven frameworks combined with compositional techniques for stability certification and controller design of \emph{finite-dimensional} interconnected networks. Assuming that subsystem dynamics are known, \citet{zhang2023compositional} adopt a compositional framework to construct neural certificates for networked dynamical systems, leveraging ISS Lyapunov functions and associated controllers of the underlying subsystems. Data-driven approaches for stability certification based on ISS and incremental ISS properties of subsystems using \emph{multiple trajectories} to gather two consecutive state-pair data are proposed by~\citet{lavaei2023ISS} and~\citet{zaker2025data}, respectively, while~\citet{zaker2025hscc} study the construction of ISS controllers and Lyapunov functions for unknown subsystems using a single input–state trajectory. Moreover, \citet{samari2025data} develop a data-driven approach to synthesize sliding mode controllers for nonlinear interconnected networks utilizing the ISS property of subsystems, while \citet{baggio2021data} construct optimal controllers for linear interconnected systems by considering the entire network as a single system. Although all the aforementioned works~\citep{lavaei2023ISS,zaker2025data,zaker2025hscc,samari2025data} extend subsystem-level results to the network via compositional techniques, they are limited to finite networks and do not extend to infinite networks. Moreover, these approaches assume noise-free data, whereas real-world data are invariably corrupted by noise.

Very recently, several studies have emerged providing data-driven \emph{safety} (rather than stability) guarantees for infinite networks. In this context,~\citet{aminzadeh2024compositional} propose a max-type small-gain compositional framework for verifying safety of infinite networks composed of discrete-time subsystems. Furthermore,~\citet{zaker2025infinite} develop a data-driven compositional condition that leverages Lipschitz bounds on subsystem dynamics together with storage certificates constructed for each subsystem. While promising, the approaches by~\citet{aminzadeh2024compositional} and~\citet{zaker2025infinite} rely on gridding the state space to collect data, which leads to \emph{exponential} sample complexity within each subsystem and often requires thousands of samples per subsystem to demonstrate their results. Moreover, both works assume \emph{noise-free} data and focus on \emph{verification} rather than controller synthesis, which is inherently a more challenging task.

\subsection{Our Contributions}

The primary objective of this work is to develop a data-driven framework for constructing a uniformly globally asymptotically stabilizing (UGA-stabilizing) controller together with an associated control Lyapunov function (CLF) for infinite networks of unknown continuous-time input-affine nonlinear subsystems. Specifically, {we first design an ISS Lyapunov function and a corresponding ISS controller for each individual subsystem directly from data, and then lift these results to the network level through a compositional strategy grounded in the small-gain theorem for infinite networks}. The main contributions of this paper, compared to the state of the art, are summarized as follows:
\begin{enumerate}
	\item[\textsl{(i)}] To the best of our knowledge, this work is the first to investigate data-driven UGA-stabilizing controller synthesis for unknown \emph{infinite} networks. More precisely, it proposes a {direct} data-driven approach based solely on a single input–state trajectory, \emph{i.e.}, {non-independent and identically distributed (non-i.i.d.)} data. We re-emphasize that the approach by~\citet{zaker2025hscc} is restricted to finite networks and, as such, cannot address the problem considered in this work.
	
	\item[\textsl{(ii)}] We assume that exact state-derivative measurements are unavailable and are instead corrupted by measurement noise with known bounds~(cf. Assumption~\ref{asmp:noise assumption}). Under this setting, unlike the one considered by~\citet{zaker2025hscc}, which assumes noise-free data, we construct ISS Lyapunov functions and corresponding ISS controllers for the subsystems directly from the noisy data; see Theorems~\ref{thm:main} and~\ref{thm:main 2}.
	
	\item[\textsl{(iii)}] We consider input-affine nonlinear polynomial dynamics that allow for {state-dependent} control input matrices (see Definition~\ref{def:ct-NPS}), unlike the setting considered by~\citet{zaker2025hscc}, which restricts attention to constant ones.
	
	\item[\textsl{(iv)}] We assume that the adversarial matrices representing the interconnection weights (matrices $D_i$ in~\eqref{eq:subsystem}) are unknown, and that only an upper bound on their induced $2$-norm is available (cf. Assumption~\ref{asmp:D}), unlike~\citet{zaker2025hscc} which assume these matrices are fully known. Under this assumption, we establish the main result stated in Theorem~\ref{thm:main}.
	
	\item[\textsl{(v)}] Finally, unlike the works by~\citet{aminzadeh2024compositional} and~\citet{zaker2025infinite}, which rely on \emph{exponential} state-space gridding to collect data and assume \emph{noise-free} measurements while focusing on the \emph{verification} problem, our paper develops a data-driven \emph{controller synthesis} framework using only a \emph{single noise-corrupted} input–state trajectory from each subsystem.
\end{enumerate}

To illustrate the capability of the proposed data-driven framework, we examine three distinct infinite networks with both sparse and dense interconnection structures, thereby demonstrating the effectiveness of our approach. For comparison, the results are summarized in a table that highlights the performance of our methodology in terms of connectivity degree, memory usage, and runtime.

\subsection{Organization}
The remainder of the paper is organized as follows. Section~\ref{sec:Problem Description} presents formal definitions of the subsystem and infinite network dynamics, and the notions of UGAS and ISS. Section~\ref{sec:ISS Lyapunov} describes the data collection procedure and develops the main theoretical data-driven results, which establish the ISS property for subsystems with unknown polynomial dynamics. Section~\ref{sec:compose} derives small-gain-based compositional results for constructing a network-level CLF and synthesizing a controller that renders the infinite network UGAS. Section~\ref{sec:simulation} reports simulation studies, demonstrating the effectiveness of the proposed approach, while Section~\ref{sec:conclusion} concludes the paper.

\subsection{Notation}
We use $\R$, $ \R_0^+$, and $\R^+$ to denote the sets of all real numbers, non-negative real numbers, and positive real numbers, respectively. Moreover, the set of positive integers is denoted by $\N^+ := \{1, 2, \ldots\}$. The cardinality of a set $\mathbf M$ is denoted by {$\mathtt{Card}(\mathbf M)$}.
A vector $x$ composed of a {countably infinite number} of vectors $x_i\in\R^{n_i}$ is represented by $x=(x_i)_{i\in \N^+}$. A block matrix $A$ with infinitely many block entries $A_{ij}$ is denoted by $A=(A_{ij})_{i,j\in\N^+},$ while a block diagonal matrix $B$ with a countably infinite number of matrix entries is represented by $B = \Diag{i}{B_i}$. We denote \emph{vector} norms on both \emph{finite- and infinite-dimensional} spaces by $|\cdot|$, and their corresponding \emph{induced} operator norms by $\|\cdot\|$. 
Whenever the argument of a norm is indexed by $i$, the norm refers to the \emph{finite}-dimensional norm associated with the $i$-th subsystem; otherwise, the norm is understood to be taken over the {infinite}-dimensional network space. For a finite-dimensional vector $x_i=[x_{1_i}~x_{2_i}~\ldots~x_{n_i}]^\top\in\R^{n_i}$, the $p$-norm is denoted by $|x_i|_p$ and defined as $|x_i|_p=(\sum_{j = 1}^{n_i}|x_{j_i}|^p)^{\sfrac{1}{p}}$. When no subscript is specified, $|\cdot|$ denotes the absolute value.
We use the notation $\ell^p$ to refer to the Banach space that includes all real vectors $x=(x_i)_{i\in \N^+}$, with finite $\ell^p$-norm $\vert x\vert_p < \infty$, where
$\vert x\vert_p = (\sum_{i = 1}^{\infty} \vert x_i\vert_p ^p)^{\sfrac{1}{p}},$ of which $\ell^2$ is specifically used to denote the Hilbert space {endowed with the canonical inner product. The space $\ell^\infty$ is a Banach space consisting of sequences equipped with the norm $\lvert x\rvert_\infty := \sup_{i\in\N^+}\lvert x_i\rvert$.} 
Given sets $X_i, i\in \N^+$, their Cartesian product is denoted by $\prod_{i\in \N^+}\!X_i$. For any {symmetric} matrix $P$, its smallest and largest eigenvalues are denoted by $\lambda_{\min}(P)$ and $\lambda_{\max}(P)$, respectively. {A matrix $P \in \R^{n \times n}$ is said to be positive definite (respectively, positive semi-definite), denoted by $P \succ 0$ (respectively, $P \succeq 0$), if it is {symmetric} and all its eigenvalues are strictly positive (respectively, non-negative). For $P \succ 0$, we denote by $\sqrt{P}$ the (unique) positive-definite matrix $A$ such that $A^2=P$.} The transpose of a matrix $P$ is represented by $P^\top$. We denote a transposed element in a symmetric position of a symmetric matrix by $\star$. The identity matrix of size $n$ is denoted by $\mathds{I}_n$. The zero matrix of dimension $n \times m$ is denoted by $\bbzero_{n \times m}$, while $\bbzero_n$ denotes the zero vector of dimension $n$. The horizontal stacking of vectors $x_i \in \R^n$ into an $n \times N$ matrix is compactly written as $[x_1 ~ x_2 ~ \dots ~ x_N]$. Given a matrix $P$, we denote by $\operatorname{rank}(P)$ and $\mathscr{R}(P)$ its rank and range, respectively.

We define function classes based on their structural properties. {A continuous function $\beta\!: \R_0^+ \rightarrow \R_0^+$ belongs to class $\mathcal{K}$ if it is strictly increasing and satisfies $\beta(0) = 0$.} If such a function also satisfies $\beta(r) \rightarrow \infty$ as $r \rightarrow \infty$, it is said to be of class $\mathcal{K}_\infty$. A function $\beta\!: \R_0^+ \times \R_0^+ \rightarrow \R_0^+$ belongs to class $\mathcal{KL}$ if, for each fixed $s$, the mapping $r \mapsto \beta(r, s)$ is a $\mathcal{K}$-function, and for each fixed $r > 0$, the function $s \mapsto \beta(r, s)$ is strictly decreasing and converges to zero as $s \to \infty$. 

\section{Problem Description}
\label{sec:Problem Description}

\subsection{Subsystem Description}

This work is concerned with infinite networks comprising continuous-time input-affine nonlinear polynomial subsystems. We begin by formally defining such subsystems.
\begin{definition}\label{def:ct-NPS}
	Let $i\in\N^+$ be given.
	A {continuous-time input-affine nonlinear polynomial subsystem (\NPS)} is described by 
	\begin{align}\label{eq:subsystem}
		\Sigma_i\!: \dot x_i=f_i(x_i) + g_i(x_i)u_i + D_iw_i,      
	\end{align}
	where $f_i: \R^{n_i}\to\R^{n_i}$ is a vector polynomial function in states $x_i\in \R^{n_i}$, satisfying $f_i(\bbzero_{n_i})=\bbzero_{n_i}$, $u_i\in \R^{m_i}$ is the control input
	to be designed as a \emph{continuous} state-feedback law, and $g_i: \R^{n_i}\to \R^{n_i\times m_i}$ is the control matrix polynomial function. This implies that $f_i$ and $g_i$ are locally Lipschitz continuous {and thus Lipschitz continuous on bounded balls}.
	Here,
	\begin{subequations}\label{eq:interconnection}
		\begin{align}
			&w_i =(w_{ij})_{j\in\mathbf M_i}\in\R^{\sigma_i}=\prod_{j\in \mathbf M_i}\R^{n_j},\label{eq:w partition}\\
			&D_i = (D_{ij})_{j\in\mathbf M_i}\in\R^{n_i\times \sigma_i},\label{eq:D partition}
		\end{align}
	\end{subequations}
	are the adversarial input and matrix, respectively, with $\mathbf M_i \subset \N^+$ being the \emph{finite} set of subsystems affecting $\Sigma_i$
	and $\sigma_i = \underset{j\in\mathbf M_i}{\sum}n_j$. We assume that the dimension of $w_{ij}$ matches that of $x_j,$ for any $i\in\N^+, j\in\mathbf M_i$. Note that $i\notin \mathbf M_i,$ for any $i\in\N^+,$ and $w_{ij}\in\R^{n_j}$. For the sake of a simpler representation, we define 
	\[
	\mathscr F_i(x_i, w_i, u_i)\Let f_i(x_i) + g_i(x_i)u_i + D_iw_i.
	\]
\end{definition}
To ensure the existence of a solution, and therefore, the well-posedness of the subsystems, we make the following assumption.
\begin{assumption}
	\label{asmp:uniform lip}
	We assume that $f_i$ and $g_i$ are locally Lipschitz continuous \emph{uniformly in $i$}, that is, for every $\varphi>0$ there exists $L>0$ such that for all $x_i,y_i \in\R^{n_i}$ with $|x_i|_2\leq \varphi$, $|y_i|_2\leq \varphi$, one has   
	\begin{subequations}
		\begin{align}   
			|f_i(x_i)-f_i(y_i)|_2&\leq L|x_i-y_i|_2,\\
			{\|g_i(x_i)-g_i(y_i)\|_2}&\leq L|x_i-y_i|_2.
		\end{align}
	\end{subequations}
\end{assumption}
This assumption is mainly required to show the well-posedness of the infinite network (cf. Lemma~\ref{lem:lip}).

Without affecting generality, one can reformulate the dynamics in~\eqref{eq:subsystem} as
\begin{align}\label{eq:subsystem reform}
	\Sigma_i\!: \dot x_i = A_i^*\mathcal F_i^*(x_i) + B_i^*\mathcal G_i^*(x_i)u_i + D_iw_i,
\end{align}
where $A_i^*\in\R^{n_i\times N_i^*}$ is the constant system matrix and $\mathcal F_i^*:\R^{n_i}\to\R^{N_i^*}$ is a vector function consisting of monomials in states,
while $B_i^*\in \R^{n_i\times M_i^*}$ is the constant control matrix and $\mathcal G_i^*:\R^{n_i}\to\R^{M_i^*\times m_i}$ is a matrix function consisting of monomials in states. Throughout this work, all matrices $A_i^*, B_i^*,$ and $D_i$, along with the exact monomials in $\mathcal F_i^*(x_i)$ and $\mathcal G_i^*(x_i)$ are assumed to be unknown. However, $\mathcal F_i^*(x_i)$ and $\mathcal G_i^*(x_i)$ satisfying 
\[
\mathcal F_i^*(x_i) = \Upsilon_{\mathcal F_i}\mathcal F_i(x_i), \qquad \mathcal G_i^*(x_i) = \Upsilon_{\mathcal G_i}\mathcal G_i(x_i),
\]
respectively, where $\mathcal F_i:\R^{n_i}\to\R^{N_i}$
and $\mathcal G_i:\R^{n_i}\to\R^{M_i\times m_i}$ are, respectively, \emph{known} vector and matrix-valued functions with monomial components, called \emph{dictionaries of monomials}, 
for some \emph{unknown} matrices $\Upsilon_{\mathcal F_i}\in\R^{N_i^*\times N_i}$ and $\Upsilon_{\mathcal G_i}\in\R^{M_i^*\times M_i},$ leading to
\begin{align}
	\label{eq:subsystem reform 1}
	\Sigma_i\!: \dot x_i = A_i\VecF(x_i) + B_i\mathcal G_i(x_i)u_i + D_iw_i,
\end{align}
with $A_i\Let A_i^*\Upsilon_{\mathcal F_i}\in\R^{n_i\times N_i}$ and $B_i\Let B_i^*\Upsilon_{\mathcal G_i}\in\R^{n_i\times M_i}$ being the unknown system and control matrices, respectively. We assume $N_i^*\!\leq\! N_i$ and $M_i^*\!\leq\! M_i$, which, together with the existence of $\Upsilon_{\mathcal F_i}$ and $\Upsilon_{\mathcal G_i}$, ensure that the dictionaries $\mathcal F_i(x_i)$ and $\mathcal G_i(x_i)$ contain all terms in $\mathcal F_i^*(x_i)$ and $\mathcal G_i^*(x_i)$, possibly augmented with additional terms to guarantee their completeness. Furthermore, the available dictionary 
$\mathcal F_i(x_i)$ satisfies $\VecF(\bbzero_{n_i})=\bbzero_{N_i}$ and
\begin{align}\label{eq:trans x}
	\VecF(x_i) = \Psi_i(x_i)x_i,
\end{align}
with $\Psi_i(x_i)\in\R^{N_i\times n_i}$ being referred to as the transformation matrix, which contains monomials, possibly with degrees from zero up to (but excluding) the maximum monomial degree.
Consequently, the \NPS\ in~\eqref{eq:subsystem reform} can be expressed as
\begin{align}
	\label{eq:subsystem reform 2}
	\Sigma_i: \dot x_i = A_i\Psi_i(x_i)x_i + B_i\mathcal G_i(x_i)u_i + D_iw_i.
\end{align}
We employ the tuple {$\Sigma_i\!=\!(A_i,B_i, D_i, \mathcal F_i, \mathcal G_i)$} to denote the subsystem in~\eqref{eq:subsystem reform 2}.

\begin{remark}
	We note that since $\mathcal F_i (\bbzero_{n_i}) = \bbzero_{N_i}$, condition~\eqref{eq:trans x} can be satisfied without loss of generality by always finding a matrix $\Psi_i(x_i)$ that satisfies it. This condition, in turn, helps ensure that all derivations are eventually represented in terms of $x_i$ rather than $\mathcal F_i (x_i)$, aligning with the choice of ISS Lyapunov functions and facilitating the proofs of our main results (cf. Theorems~\ref{thm:main} and~\ref{thm:main 2}). We also note that assuming the availability of sufficiently rich dictionaries $\mathcal F_i(x_i)$ and $\mathcal G_i(x_i)$ is common in the literature and not overly restrictive. In many engineering systems, partial insight into the underlying dynamics can typically be obtained from first principles, for instance, upper bounds on the maximum degrees of $\mathcal F_i^\ast(x_i)$ and $\mathcal G_i^\ast(x_i)$ {are information that can be obtained from such insights,} which allows $\mathcal F_i(x_i)$ and $\mathcal G_i(x_i)$ to be constructed by including all state monomials up to these known bounds.
\end{remark}

\subsection{Infinite Networks}

Having introduced individual subsystems $\Sigma_i\!=\!(A_i,B_i,$ $D_i,\mathcal F_i,\mathcal G_i)$,
we now define an infinite network induced by $\Sigma_i$.

\begin{definition}\label{def:inf net}
	Consider subsystems {$\Sigma_i\!=\!(A_i,B_i, D_i,$ $\mathcal F_i,\mathcal G_i)$}, $i\in\N^+$, as defined in Definition~\ref{def:ct-NPS}, with interconnection configuration in~\eqref{eq:interconnection} and the interconnection constraint
	\begin{align}\label{eq:int constraint}
		w_{i j}=x_j,\quad \forall i\in\N^+, \quad \forall j\in \mathbf M_i,
	\end{align}
	{\ie, the states $x_j$, $j\in \mathbf M_i$, are treated as the adversarial inputs into the subsystem $\Sigma_i$.} 
	The infinite network $\Siginf$ is defined by
	\begin{align}\label{eq:inf network}
		{\Sigma:\quad \dot x={f(x)} + g(x) u.}
	\end{align}
	{Here, $f(x)\Let A(x)x$, ${A(x)\Let (A_{ij}(x_i))_{i,j\in\N^+}}$, which satisfies $f(\pmb{\bbzero})=\pmb{\bbzero}$, with $\pmb{\bbzero}$ being an infinite vector of zero elements and}
	\begin{align*}
		x &= (x_i)_{i\in\N^+}\in X ,\ u = (u_i)_{i\in\N^+}\in U,\\
		{A_{ij}(x_i)}& \coloneq \begin{cases}
			A_i\Psi_i(x_i), & i\in\N^+,\;j=i,\\
			D_{ij}, & i\in\N^+,\; j\in\mathbf M_i,\\
			\bbzero_{n_i\times n_j}, &  \text{otherwise,}
		\end{cases}\\
		{g(x)} & \coloneq \Diag{i}{B_i\mathcal G_i(x_i)},
	\end{align*}
	where {the state and input spaces for the network are defined as the $\ell^2$ Hilbert spaces of sequences}
	\begin{align*}
		X &= \Big\{x\,\big |\, x_i \!\in\! \R^{n_i}, \vert x\vert_2 \!=\! \Big(\sum_{i = 1}^{\infty} \vert x_i\vert_2 ^2\Big)^{\sfrac{1}{2}}\!<\infty \Big\}\subset\prod_{i\in\N^+}\R^{n_i},\\
		U &= \Big\{u \,\big |\, u_i \!\in\! \R^{m_i}, \vert u\vert_2 \!=\! \Big(\sum_{i = 1}^{\infty} \vert u_i\vert_2 ^2\Big)^{\sfrac{1}{2}}\!<\infty \Big\}\subset\prod_{i\in\N^+}\R^{m_i}.\\
	\end{align*}
	We denote the infinite network in~\eqref{eq:inf network} by the tuple $\Siginfset$.
	We define 
	\[
	\mathscr F(x, u)\Let f(x) + g(x) u
	\]
	to simply refer to the network dynamics,
	and it is clear that $\mathscr F(x,u) = (\mathscr F_i(x_i, w_i, u_i))_{i\in\N^+}$.
\end{definition}

{The main objective of this work is to design a (nonlinear) feedback controller for the infinite network $\Siginfset$ such that the closed-loop system is uniformly globally asymptotically stable.} {We define the space of admissible feedback as}
{\begin{align*}
		&\mathcal U \!= \!\Big\{\!u\!: X \to U \,\big| \\
		&\hphantom{\mathcal U \!= \!\Big\{\!u\!:  X}u \text{ is Lipschitz continuous on bounded balls}\!\Big\}.
\end{align*}}

By picking an admissible feedback $u\Let K(x)\!\in\! \mathcal U$, the corresponding closed-loop system takes the form
\begin{align}\label{eq:inf network-closed-loop}
	{\Sigma_K:\quad \dot x={f(x)} + g(x) K(x).}
\end{align}
The closed-loop infinite network of interconnected subsystems is denoted by $\SiginfK$ and its tuple is represented by $\SiginfsetK$. Let $\xi:I\!\to\! X$ be a continuous mapping defined on an interval $I=[0,\mathcal T_*)$ with $\mathcal T_*\!\in\!(0,\infty]$. We say that $\xi$ is a \emph{solution} of the {closed-loop} infinite network~\eqref{eq:inf network-closed-loop} corresponding to the initial state $x_0\!\in\! X$ if
\begin{align}\label{eq:Riemann}
	{\xi(t)=x_0+\int_0^t \mathscr F(\xi(s),K(\xi(s)))\,ds, \quad t\in I}
\end{align}
is satisfied. Since the integrand is continuous {(cf. Lemma~\ref{lem:lip})}, the integral in~\eqref{eq:Riemann} is well-defined as {an $X$-valued Riemann integral}.

The {closed-loop} network $\SiginfsetK$
is said to be \emph{well-posed} if, for every $x_0\in X$, there exists a unique maximal solution $\xi$ satisfying~\eqref{eq:Riemann}, which we denote by $\phi_K(t,x_0)$  and its maximal interval of existence is represented by $I_{\max}(x_0)$. If $I_{\max}(x_0)=\R^+$ holds for all $x_0\in X$, the network is called \emph{forward complete}.

{The following lemma establishes that the dynamics of the infinite network are well-posed.
	\begin{lemma}\label{lem:lip}
		Suppose that Assumption~\ref{asmp:uniform lip} holds. Then the {closed-loop} network is well-posed.
	\end{lemma}
	
\textbf{Proof.} As Assumption~\ref{asmp:uniform lip} holds, 
		the map $\mathscr F(x,u)$ is Lipschitz continuous with respect to $x$ on bounded balls, that is, for any $\varphi>0$ there exists  $L>0$ such that whenever $|x_1|_2\leq \varphi$, $|x_2|_2\leq \varphi$, $|u|_2\leq \varphi$, one has 
		$ \vert \mathscr F(x_1,u) - \mathscr F(x_2,u)\vert_2 \leq L \vert x_1 - x_2\vert_2.$ 
		
		\noindent This guarantees that for every initial condition and admissible input, the closed-loop system $\Sigma_K$ possesses a unique solution and is therefore well-posed.
}

Having established existence and uniqueness of solutions for the closed-loop infinite network $\SiginfsetK$, we present sufficient conditions for uniform global asymptotic stability  of the infinite network, leveraging the subsystem structure and the interconnection topology, as detailed in the next subsection.

\subsection{UGAS Property}

We present the formal notion of uniform global asymptotic stability for the infinite network $\SiginfK$ in the following definition, adapted from~\citet{Mironchenko2023ISS}.

\begin{definition}\label{def:UGAS}
	The closed-loop infinite network $\SiginfK$ is called \emph{uniformly globally asymptotically stable (UGAS)}\footnote{Throughout the paper, the abbreviation UGAS is used interchangeably for uniform global asymptotic stability and uniformly globally asymptotically stable, with the intended meaning being clear from the context.} if there exists $\beta\in\mathcal{KL}$ such that
	\begin{align*}
		\vert \phi_K(t,x_0) \vert_2 \leq \beta(\vert x_0\vert_2 ,t),\quad t\geq 0,\quad x_0\in X.
	\end{align*}
\end{definition}

{\begin{remark}
		For ODEs with Lipschitz continuous right-hand sides, UGAS is equivalent to the classical global asymptotic stability (global attractivity together with local stability), see, \eg, \citep[Theorem B.37]{Mironchenko2023ISS}. 
		For infinite networks, however, UGAS is significantly stronger than global asymptotic stability, see, \eg, \citep[Examples 6.17, 6.22]{Mironchenko2023ISS}.
	\end{remark}
	
	As our aim is stabilization, we additionally define the uniform global asymptotic stabilizability for the infinite network $\Siginf$ in the following definition.
	
	\begin{definition}
		\label{def:UGA-stabilizable}
		The infinite network $\Siginf$ is called \emph{uniformly globally asymptotically stabilizable (UGA-stabilizable)} if there exists $u=K(\cdot)\in\mathcal U$ such that $\Sigma_K$ is UGAS.
	\end{definition}
}
For a continuous function $\mathcal V: X \to \R^+_0$, we define the  \emph{Lie derivative} $\dot{\mathcal V}(x_0)$ (along $\phi_K$) as the \emph{upper right-hand Dini derivative} at zero for the function $t \mapsto \mathcal V(\phi_K(t,x_0))$:
\begin{align}\label{eq:dini V}
	\dot{\mathcal V}(x_0)    &\coloneq \limsup_{s \to 0} \tfrac{1}{s}\big(\mathcal V(\phi_K(s,x_0)) - \mathcal V(x_0)\big).
\end{align}
We now present the following theorem, which elucidates the sufficient conditions for $\Siginf$ to achieve UGAS, as specified in Definition~\ref{def:UGAS}~\citep{kawan2020lyapunov}.

\begin{theorem}\label{thm:CLF}
	Given an infinite network $\Sigma\!=\!(f,g,X,U)$, assume there exist a function $\mathcal V\!:X\to\R_0^+$, known as a control Lyapunov function (CLF), a feedback $u=K(\cdot)\in\mathcal U$, 
	and constants $ \underline{\alpha},\overline{\alpha},\kappa >0,$ satisfying
	\begin{equation}
		\label{eq:con1-network}
		\underline{\alpha}\vert x \vert_2^2\le \V\le \overline{\alpha}\vert x \vert_2^2, \quad \quad {\forall} x\!\in\! X,
	\end{equation}
	and
	\begin{align}
		\label{eq:con2-network}
		\Vp \leq -\kappa \V,\quad \quad {\forall} x\!\in\! X.
	\end{align}
	Then, the infinite network $\SiginfK$ is UGAS in the sense of Definition~\ref{def:UGAS}.
\end{theorem}

Establishing a CLF guaranteeing an infinite network is UGAS typically entails substantial computational effort, even when the system model is available. A practical alternative is to examine the network’s stability by exploiting the \emph{input-to-state} stability property of its subsystems, as formalized in the following definition~\citep{Son08}, while subsequently applying a small-gain compositional reasoning.

\begin{definition}
	Given a \NPS\ $\Sigma_i\!=\!(A_i,B_i, D_i,$ $\mathcal F_i, \mathcal G_i)$, a smooth function $\mathcal V_i\!:\!\R^{n_i}\to\R_0^+$ is
	called an \emph{ISS Lyapunov function} if there exist constants $\underline{\alpha}_i, \overline{\alpha}_i,\kappa_i \in \R^+,$ and $\rho_i \in \R^+_0,$ such that
	\begin{itemize}
		\begin{subequations}
			\item  $\forall x_i\!\in\R^{n_i}\!\!:$
			\begin{align}\label{eq:ISS-con1}
				\underline{\alpha}_i |x_i|_2^2\le \!\Vs \le \overline{\alpha}_i |x_i|_2^2,
			\end{align}
			\item $\forall x_i\in\R^{n_i},\: \exists u_i\in \R^{m_i},$ such that $\forall w_i\in \R^{\sigma_i}\!\!:$
			\begin{align}\label{eq:ISS-con2}
				&\dot{\mathcal V}_i(x_i) ~\leq - \kappa_i \Vs +  \rho_i |w_i|_2^2,
			\end{align}
		\end{subequations}
	\end{itemize}
	where $\dot{\mathcal V}_i$ is the Lie derivative of $\mathcal V_i:\R^{n_i}\to\R_0^+$ with respect to the dynamics in~\eqref{eq:subsystem reform 1} (or equivalently \eqref{eq:subsystem reform 2}), and since $\mathcal V_i$ is smooth, it can be computed as 
	\begin{align}\label{Lie derivative}
		\dot{\mathcal V}_i(x_i)=\partial_{x_i}\! \mathcal V_i(x_i)(A_i\mathcal F_i(x_i) + B_i\mathcal G_i(x_i)u_i + D_iw_i),
	\end{align}
	with $\partial_{x_i} \!\Vs = \frac{\partial \Vs}{\partial x_i}$.
\end{definition}

Although a CLF for the infinite network $\Siginf$ can, in principle, be synthesized by combining the ISS Lyapunov functions of the individual subsystems through a compositional condition, the construction of such ISS Lyapunov functions is infeasible due to the unknown matrices $A_i$, $B_i$, and $D_i$ in \eqref{Lie derivative}. Motivated by this key challenge, we next state the primary problem of interest addressed in this work.

\begin{resp}
	\begin{problem}\label{prob}
		Consider an infinite network $\Siginf$ consisting of countably infinitely many subsystems {$\Sig,$} with \emph{unknown} matrices $A_i, B_i,$ and $D_i$. Develop a direct data-driven framework to synthesize an ISS Lyapunov function $\mathcal V_i$ along with its corresponding controller $u_i,$ utilizing a \emph{single noisy} input–state trajectory collected from each subsystem. Subsequently, provide a small-gain compositional technique for constructing a CLF $\mathcal V$ composed of data-driven $\mathcal V_i$ and a decentralized UGA-stabilizing controller $u$ built from the local controllers $u_i$, each designed purely from data, that guarantees the infinite network is UGAS.
	\end{problem}
\end{resp}

We begin addressing Problem~\ref{prob} by presenting our data-driven framework in the following section.

\section{Data-Driven Design of ISS Lyapunov Functions} \label{sec:ISS Lyapunov}

\subsection{Data Collection}

Given an open-loop infinite network $\Siginfset$ with interconnection constraint~\eqref{eq:int constraint}, 
we start by collecting $T \in \mathbb{N}^+$ samples from the unknown \NPS\ in~\eqref{eq:subsystem reform 2}, with sampling time $\tau\in\R^+$, starting from fixed initial time $t_0$ and over the period $[t_0, t_0+(T - 1)\tau]$. We apply arbitrary admissible inputs $u_i$ over the time interval $[t_0, t_0+(T - 1)\tau]$ to each subsystem initialized at $x_i(t_0)$, and store $\mathbf U_i$ and $\mathbf X_i$. Hence, considering~\eqref{eq:data-unknown-derivatives}, the following data can be collected:
\begin{align}\label{eq:data}
	\begin{array}{llllll}
		\mathbf U_i & \hspace{-0.1cm}= & \hspace{-0.1cm}[u_i(t_0) &u_i(t_0 + \tau) & \dots & u_i(t_0 +(T-1)\tau)] \! \in \! \R^{m_i\times T}\!\!,\\
		\mathbf W_i & \hspace{-0.1cm}= & \hspace{-0.1cm}[w_i(t_0) & w_i(t_0 + \tau) & \dots & w_i(t_0 +(T-1)\tau)] \! \in \! \R^{\sigma_i\times T}\!\!,\\
		\mathbf X_i & \hspace{-0.1cm}= & \hspace{-0.1cm}[x_i(t_0) & x_i(t_0 + \tau) & \dots & x_i(t_0 + (T-1)\tau)] \! \in \! \R^{n_i\times T}\!\!.
	\end{array}
\end{align} 
We note that since the interconnection topology among the subsystems is known as per~\eqref{eq:int constraint} and the state data $\mathbf X_i$ of all subsystems are collected, $\mathbf W_i$ can be obtained as  $\mathbf W_i=(\mathbf W_{ij}) _{j\in \mathbf M_i}\in\R^{\sigma_i\times T}$ with
\begin{align}\label{eq:W_i_int constraint}
	\mathbf W_{i j}= \mathbf X_j,\quad \forall i\in\N^+, \quad \forall j\in \mathbf M_i.
\end{align}
The sampled data in~\eqref{eq:data} are referred to as a {set of input–state trajectories}.

{Let the state derivative data be denoted by}
\begin{align}\label{eq:data-unknown-derivatives}
	\begin{array}{llllll}
		\hat{\mathbf X}_i^{\mathsf d} & \hspace{-0.1cm}= & \hspace{-0.1cm}[\dot x_i(t_0) & \dot x_i(t_0 + \tau) & \dots & \dot x_i(t_0 + (T-1)\tau)]\in\R^{n_i\times T}\!\!.
	\end{array}
\end{align} 
However, {the state derivatives in $\hat{\mathbf X}_i^{\mathsf d}$ cannot be} directly measured. Hence, we numerically approximate their values by forward differences
\begin{align}\label{Newd}
	\dot x_i(t_0 \!+\! k\tau)\! =\! \frac{x_i(t_0 \! + \! (k \! + \! 1)\tau) \!-\! x_i(t_0 + k\tau)}{\tau} \!+\! \delta_i(t_0 \!+\! k\tau),
\end{align}
for $k\in\{0, \dots, T-1\}$, with $\delta_i(t_0 + k\tau)$ being the approximation error (similar to the higher-order terms appearing in Taylor's expansion~\citep{guo2024data}), which can be considered as \emph{measurement noise}~\citep{guo2021data}. We note that we collect one extra sample in $\mathbf X_i$, {\ie, $T+1$ samples,} solely for the calculation of the state derivatives {when $k=T-1$, which requires $x_i(t_0 + T\tau)$,} while we use $T$ samples in our data-driven analysis. Hence, according to~\eqref{Newd}, the {state derivative data satisfy}
\[
\hat{\mathbf X}_i^{\mathsf d} = \mathbf X_i^{\mathsf d} + \Delta_i, \quad \text{where}
\]
\[
\Delta_i = [\delta_i(t_0)~\delta_i(t_0+\tau)~\dots~\delta_i(t_0 + (T-1)\tau)]\in\R^{n_i\times T}
\] 
represents the unknown noise corrupting the measurement, and $\mathbf X_i^{\mathsf d}$ defined by
\[
(\mathbf X_i^{\mathsf d})_k \!=\! \frac{x_i(t_0+(k+1)\tau) \!-\! x_i(t_0 + k\tau)}{\tau},~ \!k\!=\!0,\ldots,T-1,
\]
is {the available state derivative approximation for analysis.} In this regard, we make the following assumption.
\begin{assumption}\label{asmp:noise assumption}
	The measurement noise trajectory $\Delta_i\in\R^{n_i\times T}$ is unknown, but it satisfies $\Delta_i\Delta_i^\top\preceq\Lambda_i\Lambda_i^\top$\!, for some known $\Lambda_i$ of appropriate dimension. 
\end{assumption}
{The upper bound in Assumption~\ref{asmp:noise assumption} indicates that the noise energy is bounded~\citep{van2020noisy}, which is reasonable since it accounts for the error in approximating the state derivative.}
\begin{remark}\label{rem:noise}
	A practically relevant situation in which the bound $\Delta_i\Delta_i^\top \preceq \Lambda_i \Lambda_i^\top$ is guaranteed arises when there exists a known constant $\bar{\delta}>0$ such that each column $\delta_{i,k}$ of $\Delta_i$ satisfies $|\delta_{i,k}|_2^2 \le \bar{\delta}$ for all $k\in\{1,\ldots,T\}$. Accordingly, for any $y \in \R^{n_i}$, using $\Delta_i\Delta_i^\top=\sum_{k=1}^{T}\delta_{i,k}\delta_{i,k}^\top$ and the Cauchy-Schwarz inequality, one has
	\begin{align*}
		y^\top \Delta_i\Delta_i^\top y
		&= \sum_{k=1}^{T} y^\top \delta_{i,k}\delta_{i,k}^\top y
		= \sum_{k=1}^{T}(\delta_{i,k}^\top y)^2 \\
		&\le \sum_{k=1}^{T}|\delta_{i,k}|_2^2\,|y|_2^2
		\le \sum_{k=1}^{T}\bar{\delta}\,|y|_2^2
		= \bar{\delta}T\,|y|_2^2 \\
		&= y^\top\!\left(\bar{\delta}T\,\mathds I_{n_i}\right)\!y .
	\end{align*}
	Since the above inequality holds for all $y$, it follows that
	\[
	\Delta_i\Delta_i^\top \preceq \bar{\delta}T\,\mathds I_{n_i},
	\]
	which implies $\Delta_i\Delta_i^\top \preceq \Lambda_i\Lambda_i^\top$ is satisfied by choosing $\Lambda_i \coloneqq \sqrt{\bar{\delta}T}\,\mathds I_{n_i}$.
\end{remark}

Although the exact adversarial matrix $D_i$ is unknown, we impose the following assumption to proceed with the development of the first main theorem of this work.

\begin{assumption}\label{asmp:D}
	The adversarial input matrix $D_i$ is unknown; however, an upper bound on its induced $2$-norm is available, namely $\|D_i\|_2 \leq \varkappa_i$, where $\varkappa_i$ is a known positive constant.
\end{assumption}

Since $\VecF(x_i)$ and $\mathcal G_i(x_i)$ are known dictionaries, using collected data $\mathbf X_i$ and $\mathbf U_i$, one can also construct matrices
\begin{subequations}\label{eq:mon traj}
	\begin{align}
		\mathbf J_i =  \big[&\VecF(x_i(t_0)\!)~~ \VecF(x_i(t_0 + \tau)\!) ~~\dots \notag\\
		&~~~~~~~~~~~~~~\VecF(x_i(t_0 + (T-1)\tau)\!)\big]\!\in\R^{N_i\times T}\!,\\
		\mathbf G_i  =  \big[&\mathcal G_i(x_i(t_0)\!)u_i(t_0) ~~\mathcal G_i(x_i(t_0 + \tau)\!)u_i(t_0 + \tau) ~~\dots \notag\\
		&\mathcal G_i(x_i(t_0 + (T-1)\tau)\!)u_i(t_0 + (T-1)\tau)\big] \!\! \in \! \R^{M_i\times T}\!.
	\end{align}
\end{subequations}
Accordingly, the dynamics of subsystem $\Sigma_i$ in~\eqref{eq:subsystem reform 1} imply the following relation between the data matrices introduced above:
\begin{align}
	\label{eq:Sigma_i_data}
	\mathbf X_i^{\mathsf d} = \underbrace{A_i\mathbf J_i + B_i\mathbf G_i + D_i\mathbf W_i}_{\hat{\mathbf X}_i^{\mathsf d}} - \Delta_i.
\end{align}    
We concatenate the unknown matrices $A_i$, $B_i$, and $D_i$ as $S_i = [A_i~B_i~D_i]$. We also denote $R_i = [A_i~B_i]$. On this basis, we consolidate the available data matrices $\mathbf J_i$, $\mathbf G_i$, and $\mathbf W_i$  into the matrix
\[
\mathbf Q_i = [\mathbf J_i^\top~\mathbf G_i^\top~ \mathbf W_i^\top]^\top\in\R^{s_i\times T}.
\]
{With these definitions, we can rewrite \eqref{eq:Sigma_i_data} in a shortened form as}
\begin{align*}
	\mathbf X_i^{\mathsf d} = S_i \mathbf Q_i - \Delta_i,
\end{align*}
equivalently, $\Delta_i = S_i\mathbf Q_i -\mathbf X_i^{\mathsf d} = -(\mathbf X_i^{\mathsf d}-S_i\mathbf Q_i).$
Then, according to Assumption~\ref{asmp:noise assumption}, we obtain
\begin{align*}
	(\mathbf X_i^{\mathsf d}-S_i\mathbf Q_i)(\mathbf X_i^{\mathsf d}-S_i{\mathbf Q_i})^\top- \Lambda_i\Lambda_i^\top\preceq0,
\end{align*}
which results in the following {linear matrix inequality:}
\begin{align}
	\label{eq:Z_i__and__Xi_i-matrices}
	\underbrace{\begin{bmatrix}
			\mathds I_{n_i}\\
			S_i^\top
		\end{bmatrix}^{\!\!\top} \overbrace{\begin{bmatrix}
				\mathbf X_i^{\mathsf d}{\mathbf X_i^{\mathsf d}}^\top-\Lambda_i\Lambda_i^\top & -\mathbf X_i^{\mathsf d}\mathbf Q_i^\top\\
				\star & \mathbf Q_i \mathbf Q_i^\top
		\end{bmatrix}}^{\mathds Z_i}\begin{bmatrix}
			\mathds I_{n_i}\\
			S_i^\top
	\end{bmatrix}}_{\Xi_i}\preceq 0,
\end{align}
where $\mathds Z_i\in\R^{(n_i + s_i)\times (n_i + s_i)},$ with $s_i = M_i+N_i+\sigma_i$, and the symbol $\star$ is the transpose of a symmetric element of a (symmetric) matrix.

\subsection{Input-to-State Stabilizing Controllers for Subsystems}

In this section, we propose a data-driven approach for designing ISS Lyapunov functions and the corresponding stabilizing controllers for subsystems of the infinite network. 

Given a \NPS\ {$\Sig$} as defined in~\eqref{eq:subsystem reform 1},
consider a controller 
\begin{align}
	\label{eq:Controller-subsystem}
	u_i = \mathbf K_i(x_i)P_ix_i,  
\end{align}
where $\mathbf K_i:\R^{n_i} \to \R^{m_i\times n_i}$ is a matrix polynomial in $x_i$
and $P_i\in\R^{n_i\times n_i}$ is a constant matrix to be designed later.

As an ISS Lyapunov function candidate for the $i$-th subsystem, we choose a {quadratic} function
\begin{align}
	\label{eq:ISS-LF-subsystem}
	\Vs = x_i^\top P_ix_i,\quad x_i\!\in\R^{n_i},
\end{align}
with $P_i \succ 0$.

{In the following lemma, we} introduce a useful representation of the right-hand side of $\Sigma_i$ under controller $u_i$.
\begin{lemma}\label{lem:Lemma-rep}
	{The right-hand side of the closed-loop system \eqref{eq:subsystem reform 1}, under the controller designed in~\eqref{eq:Controller-subsystem}, is represented by}
	\begin{subequations}
	\begin{align}
		A_i\VecF(x_i) &+ B_i\mathcal G_i(x_i)u_i + D_iw_i\notag\\
		&=S_i\begin{bmatrix}
			\Psi_i(x_i)\\
			\mathcal G_i(x_i)\mathbf K_i(x_i)P_i\\
			\bbzero_{\sigma_i\times n_i}
		\end{bmatrix} x_i+ D_i w_i\label{eq:representation 1}\\
		&= R_i
		\begin{bmatrix}
			\Psi_i(x_i)\\
			\mathcal G_i(x_i)\mathbf K_i(x_i)P_i
		\end{bmatrix}x_i + D_i w_i.\label{eq:representation 2}
	\end{align}
		\end{subequations}
\end{lemma}

\textbf{Proof.} Applying the controller $u_i = \mathbf K_i(x_i)P_ix_i$ to the dynamics in~\eqref{eq:subsystem reform 1}, we have
	\begin{align*}
		\dot x_i &\overset{\hphantom{\eqref{eq:trans x}}}{=} A_i \VecF(x_i) + B_i \mathcal G_i(x_i)u_i +D_iw_i\\
		&\overset{\eqref{eq:trans x}}{=} A_i\Psi_i(x_i)x_i + B_i \mathcal G_i(x_i)\mathbf K_i(x_i)P_ix_i + D_i w_i \\
		&\overset{\hphantom{\eqref{eq:trans x}}}{=} (A_i\Psi_i(x_i) + B_i\mathcal G_i(x_i)\mathbf K_i(x_i)P_i)x_i + D_i w_i.
	\end{align*}
	Now, recalling our notation for $S_i$ and $R_i$, we attain
	\begin{align*}
		\dot x_i &= \overbrace{[A_i~B_i~D_i]}^{S_i}\begin{bmatrix}
			\Psi_i(x_i)\\
			\mathcal G_i(x_i)\mathbf K_i(x_i)P_i\\
			\bbzero_{\sigma_i\times n_i}
		\end{bmatrix}x_i + D_i w_i \\
		&= S_i\begin{bmatrix}
			\Psi_i(x_i)\\
			\mathcal G_i(x_i)\mathbf K_i(x_i)P_i\\
			\bbzero_{\sigma_i\times n_i}
		\end{bmatrix}x_i + D_i w_i,\\
		&= R_i
		\begin{bmatrix}
			\Psi_i(x_i)\\
			\mathcal G_i(x_i)\mathbf K_i(x_i)P_i
		\end{bmatrix}x_i + D_i w_i,
	\end{align*}
	which concludes the proof. $\hfill\blacksquare$

In the rest of this section, we prove that the $i$-th subsystem is ISS via the controller $u_i$, with $\Vs$ as its ISS Lyapunov function, provided that $\mathbf K_i(x_i)$ and $P_i$ satisfy the following assumption.

\begin{assumption}\label{asmp:main con}
	Assume that there exist constants $\kappa_i,\vartheta_i\in\R^+$, a scalar polynomial $\gamma_i(x_i) \geq 0$, such that a matrix $P_i\succ 0$, and a matrix polynomial $\mathbf K_i$ satisfy
	\begin{align}\label{eq:main con}
		\mathds H_i(x_i) + {(\vartheta_i + \kappa_i)} \mathds P_i - \gamma_i(x_i) \mathds Z_i \preceq 0,
	\end{align}
	where $\mathds Z_i$ is given by \eqref{eq:Z_i__and__Xi_i-matrices},
	and $\mathds H_i(x_i)$ is a block matrix tailored for the representation \eqref{eq:representation 1}:
	\begin{subequations}\label{eq:main conditions}
		\begin{align}
			\mathds H_i(x_i)&= \begin{bmatrix}
				\bbzero_{n_i\times n_i} & \conblock^{\!\!\top}\\
				\star & \bbzero_{s_i\times s_i}
			\end{bmatrix}\!\!,\\
			\mathds P_i &= \begin{bmatrix}
				P_i^{-1} & \bbzero_{n_i\times s_i}\\
				\bbzero_{s_i\times n_i} & \bbzero_{s_i\times s_i}
			\end{bmatrix}\!\!,
		\end{align}
	\end{subequations}
	and $\mathds H_i(x_i),\mathds P_i\in\R^{(n_i + s_i)\times (n_i + s_i)},$ with $s_i = M_i+N_i+\sigma_i$.
\end{assumption}
\begin{remark}
	One can define $\Phi_i\Let P_i^{-1}$ as a decision variable when enforcing condition~\eqref{eq:main con}, and design $\Phi_i \succ 0$. Once $\Phi_i$ is designed, $P_i$ can be obtained by inverting $\Phi_i$, \ie, $\Phi_i^{-1} = (P_i^{-1})^{-1} = P_i$.
\end{remark}
Utilizing the collected data in~\eqref{eq:data}, the closed-loop dynamics in~\eqref{eq:representation 1}, and {under Assumptions~\ref{asmp:noise assumption}--\ref{asmp:main con},} we present our first key contribution in the following theorem, which ensures that each subsystem is rendered ISS.

\begin{theorem}\label{thm:main}
	Consider a \NPS\ {$\Sigma_i\!=\!\!(A_i,B_i, D_i,$ $\mathcal F_i,\mathcal G_i)$}, characterized by unknown matrices $A_i$, $B_i$, and $D_i$. Let Assumptions~\ref{asmp:noise assumption}--\ref{asmp:main con} hold.
	Then $u_i$ given by \eqref{eq:Controller-subsystem} renders $\Sigma_i$ ISS, and  
	$\Vs$ as defined in \eqref{eq:ISS-LF-subsystem}
	is the corresponding ISS Lyapunov function with $\underline{\alpha}_i = \lambda_{\min}(P_i)$, $\overline{\alpha}_i = \lambda_{\max}(P_i)$, $ \rho_i =  \frac{\Vert\sqrt{P_i}\Vert_2^2\varkappa_i^2}{\vartheta_i}$, for some $\vartheta_i\in\R^+$.
\end{theorem}

\textbf{Proof.}
	First, we show that condition~\eqref{eq:ISS-con1} holds for the candidate ISS Lyapunov function. Since
	\begin{align*}
		\lambda_{\min}(P_i) |x_i|_2^2 \leq \underbrace{x_i^\top P_i x_i}_{\Vs} \leq \lambda_{\max}(P_i) |x_i|_2^2, \quad {\forall} x_i \in\R^{n_i},
	\end{align*}
	$\mathcal V_i$ satisfies \eqref{eq:ISS-con1} with $\underline{\alpha}_i = \lambda_{\min}(P_i)$ and $\overline{\alpha}_i = \lambda_{\max}(P_i)$. 
	Now, we demonstrate that condition~\eqref{eq:ISS-con2} holds, as well. Utilizing the closed-loop representation~\eqref{eq:representation 1}, we have
	{\allowdisplaybreaks
		\begin{align*}
			&\Vps 
			=2x_i^\top P_i\Big(A_i\VecF(x_i) + B_i\mathcal G_i(x_i)u_i + D_iw_i\Big)\\
			&= 2x_i^\top P_i\Big(S_i\begin{bmatrix}
				\Psi_i(x_i)\\
				\mathcal G_i(x_i)\mathbf K_i(x_i)P_i\\
				\bbzero_{\sigma_i\times n_i}
			\end{bmatrix} x_i+ D_i w_i\Big)\\
			&=2x_i^\top P_iS_i\begin{bmatrix}
				\Psi_i(x_i)\\
				\mathcal G_i(x_i)\mathbf K_i(x_i) P_i\\
				\bbzero_{\sigma_i\times n_i}
			\end{bmatrix} x_i+ 2x_i^\top P_iD_iw_i \\
			&=2x_i^\top P_iS_i\conblock P_i x_i+ 2x_i^\top P_iD_iw_i\\
			&=x_i^\top P_iS_i\conblock P_i x_i\\
			&\hphantom{=}~+ x_i^\top P_i\conblock^{\!\!\top} \!\!\! S_i^\top P_i x_i+ 2x_i^\top P_iD_iw_i\\
			& = x_i^\top P_i\Big(S_i\conblock + \conblock^{\!\!\top} \!\!\! S_i^\top\Big) P_i x_i\\
			&\hphantom{=}~+ 2x_i^\top P_iD_iw_i\\
			&= x_i^\top P_i\begin{bmatrix}
				\mathds I_{n_i}\\
				S_i^\top
			\end{bmatrix}^{\!\!\top}\!\!\begin{bmatrix}
				\bbzero_{n_i\times n_i} & \conblock^{\!\!\top}\\
				\star & \bbzero_{s_i\times s_i}
			\end{bmatrix}\!\!\begin{bmatrix}
				\mathds I_{n_i}\\
				S_i^\top
			\end{bmatrix} P_i x_i\\
			&\hphantom{=}~+ 2\underbrace{x_i^\top \sqrt{P_i}\sqrt{P_i}D_iw_i}_{\text{(i)}}.
		\end{align*}
	}
	By applying the Cauchy-Schwarz inequality~\citep{bhatia1995cauchy} over $\text{(i)}$ as $a^\top b \leq |a|_2  |b|_2,$ for any $a,b\in \mathbb R^n$, with $a^\top = x_i^\top \sqrt{P_i}$ and $b = \sqrt{P_i}D_iw_i$, and by leveraging Young's inequality~\citep{young1912classes} as $|a|_2 |b|_2 \leq \frac{\vartheta_i}{2} |a|_2^2+\frac{1}{2\vartheta_i} |b|_2^2,$ for any $\vartheta_i>0$, we get
	
	\begin{align*}
		\Vps&\leq x_i^\top\! P_i\!\begin{bmatrix}
			\mathds I_{n_i}\\
			S_i^\top
		\end{bmatrix}^{\!\!\top}\!\!\!\underbrace{\begin{bmatrix}
				\bbzero_{n_i\times n_i} & \conblock^{\!\!\top}\\
				\star & \bbzero_{s_i\times s_i}
		\end{bmatrix}}_{\mathds H_i(x_i)}\!\!\!\begin{bmatrix}
			\mathds I_{n_i}\\
			S_i^\top
		\end{bmatrix}\!\! P_i x_i\\
		&\hphantom{=}~+ \vartheta_i\underbrace{x_i^\top P_i x_i}_{\Vs} + \frac{\Vert \sqrt{P_i}\Vert_2^2 \Vert D_i\Vert_2^2}{\vartheta_i} |w_i|_2^2\\
		&{\leq h_i(x_i) + \vartheta_i\Vs + \frac{\Vert \sqrt{P_i}\Vert_2^2 \varkappa_i^2}{\vartheta_i}|w_i|_2^2}\\
		& = h_i(x_i) + \vartheta_i\Vs + \rho_i|w_i|_2^2,
	\end{align*}
	where 
	\[
	h_i(x_i) =  x_i^\top P_i\begin{bmatrix}
		\mathds I_{n_i}\\
		S_i^\top
	\end{bmatrix}^{\!\!\top}\!\!\mathds H_i(x_i)\!\!\begin{bmatrix}
		\mathds I_{n_i}\\
		S_i^\top
	\end{bmatrix} P_i x_i, \quad\!\! \rho_i {=} \frac{\Vert \sqrt{P_i}\Vert_2^2 \varkappa_i^2}{\vartheta_i}.
	\]
	To show \eqref{eq:ISS-con2}, we need to ensure that
	\begin{align}\label{eq:new H}
		\underbrace{h_i(x_i) + \vartheta_i \Vs}_{\mathcal H_i(x_i)} \leq -\kappa_i \Vs,
	\end{align}
	or equivalently,
	\begin{align}\label{eq:new S}
		\mathcal H_i(x_i) + \kappa_i\Vs \leq 0
	\end{align}
	is met, which is not possible at this step since $S_i$ is unknown.
	
	\noindent Now, we invoke {our knowledge about the system based on the collected data and
		utilize \eqref{eq:Z_i__and__Xi_i-matrices}}. By multiplying $\Xi_i\in\R^{n_i\times n_i}$ by $x_i^\top P_i$ and $P_ix_i$ from left and right, respectively, we attain
	\begin{align}\label{eq:proof-new1}
		\underbrace{x_i^\top P_i \Xi_i P_i x_i}_{z_i(x_i)}\leq 0.
	\end{align}
	Then, utilizing the S-procedure~\citep{yakubovich2004stability},
	we can enforce~\eqref{eq:new S}, while respecting~\eqref{eq:proof-new1}, if there exists $\gamma_i(x_i)\geq 0$ such that
	\begin{align}\label{eq:S-procedure}
		\mathcal H_i(x_i) + \kappa_i\Vs - \gamma_i(x_i) z_i(x_i) \leq 0,
	\end{align}
	for some $\kappa_i,\vartheta_i\in\R^+$.
	Therefore, ensuring condition~\eqref{eq:main con}, with $\Vs$ reformulated as
	\begin{align*}
		\Vs& = x_i^\top P_i x_i = x_i^\top \underbrace{P_i P_i^{-1}}_{\mathds I_{n_i}}P_ix_i \\
		&= x_i^\top P_i\begin{bmatrix}
			\mathds I_{n_i}\\
			S_i^\top
		\end{bmatrix}^{\!\!\top}\!\!\underbrace{\begin{bmatrix}
				P_i^{-1} & \bbzero_{n_i\times s_i}\\
				\bbzero_{s_i\times n_i} & \bbzero_{s_i\times s_i}
		\end{bmatrix}}_{\mathds P_i}\!\!\begin{bmatrix}
			\mathds I_{n_i}\\
			S_i^\top
		\end{bmatrix} P_i x_i,
	\end{align*}
	yields \eqref{eq:S-procedure}, and thereby~\eqref{eq:new H} holds. Hence, we have
	\begin{align*}
		\Vps &\leq -\kappa_i\Vs+\rho_i |w_i|_2^2,
	\end{align*}
	which concludes the proof.$\hfill\blacksquare$

\begin{remark}
	To enforce condition~\eqref{eq:main con}, one may employ established software packages such as \textsf{SOSTOOLS}~\citep{papachristodoulou2013sostools} in conjunction with a semidefinite programming (SDP) solver, such as \textsf{SeDuMi}~\citep{sturm1999using}.
\end{remark}

As stated in Assumption~\ref{asmp:D}, exact knowledge of $D_i$ is not strictly required to compute $\rho_i = \frac{\Vert \sqrt{P_i} \Vert_2^{2} \Vert D_i \Vert_2^{2}}{\vartheta_i}$; an upper bound on $\Vert D_i \Vert_2$ suffices. Nevertheless, the dimensionality of the matrices in~\eqref{eq:Z_i__and__Xi_i-matrices} and \eqref{eq:main conditions} (\ie, matrices $\mathds H_i(x_i), \mathds P_i, \mathds Z_i$)
depends directly on the number of subsystems influencing a given subsystem, which may introduce scalability challenges when the network interconnection is dense. These difficulties can, however, be mitigated when $D_i$ is known. Since $D_i$ represents the interconnection weights among subsystems, this assumption can be realistic in certain interconnected networks. Accordingly, we propose the subsequent results, which assume that $D_i$ is known and yield more scalable conditions. {Prior to this, we introduce the following assumption, which provides the sufficient conditions for designing $\mathbf K_i(x_i)$ as well as $P_i \succ 0$.}

{\begin{assumption}\label{asmp:main con 2}
		Assume that there exist constants $\kappa_i,\vartheta_i\in\R^+$, a scalar polynomial $\gamma_i(x_i) \geq 0$, a constant matrix $P_i\in\R^{n_i\times n_i}$, where $P_i\succ 0$, and a matrix polynomial $\mathbf K_i(x_i)\in\R^{m_i\times n_i}$, such that 
		\begin{align}\label{eq:main con 2}
			\mathds C_i(x_i) + (\vartheta_i + \kappa_i) \mathds P_i - \gamma_i(x_i) \mathds Y_i \preceq 0,
		\end{align}
		where
		\begin{align*}
			\mathds C_i(x_i) &= \begin{bmatrix}
				\bbzero_{n_i\times n_i} & \conblockD^{\!\!\top}\\
				\star & \bbzero_{r_i\times r_i}
			\end{bmatrix}\!\!,\\
			\mathds P_i &= \begin{bmatrix}
				P_i^{-1} & \bbzero_{n_i\times r_i}\\
				\bbzero_{r_i\times n_i} & \bbzero_{r_i\times r_i}
			\end{bmatrix}\!\!, \\
			\mathds Y_i &= 
			\begin{bmatrix}
				{\tilde{\mathbf X}_i^{\mathsf d}\tilde{\mathbf X}_i^{\mathsf d}}^{\!\top}\!\!-\!\Lambda_i\Lambda_i^\top & -\tilde{\mathbf X}_i^{\mathsf d}\mathbf L_i^\top\\
				\star & \mathbf L_i \mathbf L_i^\top
			\end{bmatrix}\!\!,
		\end{align*}
		and $\mathds C_i(x_i),\mathds P_i,\mathds Y_i\in\R^{(n_i + r_i)\times (n_i + r_i)},$ with $r_i = M_i+N_i$, $\tilde{\mathbf X}_i^{\mathsf d} = \mathbf X_i^{\mathsf d}\!-\!D_i\mathbf W_i$, and $\mathbf L_i = [\mathbf J_i^\top~ \mathbf G_i^\top]^\top\in\R^{r_i\times T}$.
\end{assumption}}

{In the subsequent theorem, we leverage the gathered data in~\eqref{eq:data}, the closed-loop representation in~\eqref{eq:representation 2}, and Assumptions~\ref{asmp:noise assumption} and~\ref{asmp:main con 2} to ensure that each subsystem achieves ISS when $D_i$ is known.}

{\begin{theorem}\label{thm:main 2}
		Consider a \NPS\ {$\Sigma_i\!=\!(A_i,B_i, D_i,$ $\mathcal F_i,\mathcal G_i)$}, characterized by unknown matrices $A_i$ and $B_i$. Let Assumptions~\ref{asmp:noise assumption} and \ref{asmp:main con 2} hold. Then $u_i$ designed as in~\eqref{eq:Controller-subsystem} renders $\Sigma_i$ ISS, and $\Vs$ as defined in~\eqref{eq:ISS-LF-subsystem} is the corresponding ISS Lyapunov function with $\underline{\alpha}_i = \lambda_{\min}(P_i)$, $\overline{\alpha}_i = \lambda_{\max}(P_i)$, $ \rho_i =  \frac{\Vert\sqrt{P_i}\Vert_2^2\Vert D_i \Vert_2^2}{\vartheta_i}$, for some $\vartheta_i\in\R^+$.
\end{theorem}}

\textbf{Proof.}
	The first part of the proof follows the same procedure as in the proof of Theorem~\ref{thm:main}, yielding the fulfillment of condition~\eqref{eq:ISS-con1} by choosing $\underline\alpha_i = \lambda_{\min}(P_i)$ and $\overline\alpha_i = \lambda_{\max}(P_i)$. To show that condition~\eqref{eq:ISS-con2} is satisfied as well, we use the closed-loop representation~\eqref{eq:representation 2} obtained in Lemma~\ref{lem:Lemma-rep}. Following the proof steps of Theorem~\ref{thm:main}, one can obtain
	\begin{align*}
			\Vps&= x_i^\top \! P_i \!\begin{bmatrix}
				\mathds I_{n_i}\\
				R_i^\top
			\end{bmatrix}^{\!\!\top}\!\!\!\begin{bmatrix}
				\bbzero_{n_i\times n_i} & \conblockD^{\!\!\top}\\
				\star & \bbzero_{r_i\times r_i}
			\end{bmatrix}\!\!\!\begin{bmatrix}
				\mathds I_{n_i}\\
				R_i^\top
			\end{bmatrix}\!\! P_i x_i\\
			&\hphantom{=}~+ 2x_i^\top P_iD_iw_i\\
			&\leq x_i^\top\! P_i\!\begin{bmatrix}
				\mathds I_{n_i}\\
				R_i^\top
			\end{bmatrix}^{\!\!\top}\!\!\!\underbrace{\begin{bmatrix}
					\bbzero_{n_i\times n_i} & \conblockD^{\!\!\top}\\
					\star & \bbzero_{r_i\times r_i}
			\end{bmatrix}}_{\mathds C_i(x_i)}\!\!\!\begin{bmatrix}
				\mathds I_{n_i}\\
				R_i^\top
			\end{bmatrix}\!\! P_i x_i\\
			&\hphantom{=}~+ \vartheta_i\underbrace{x_i^\top P_i x_i}_{\Vs} + \frac{\Vert \sqrt{P_i}\Vert_2^2 \Vert D_i\Vert_2^2}{\vartheta_i} |w_i|_2^2\\
			& = c_i(x_i) + \vartheta_i\Vs + \rho_i |w_i|_2^2,
		\end{align*}
		where $$c_i(x_i) \!=\!  x_i^\top \!P_i\!\begin{bmatrix}
			\mathds I_{n_i}\\
			R_i^\top
		\end{bmatrix}^{\!\!\top}\!\!\!\mathds C_i(x_i)\!\!\begin{bmatrix}
			\mathds I_{n_i}\\
			R_i^\top
		\end{bmatrix} \!\!P_i x_i,\quad\!\!\! \rho_i \!= \!\frac{\Vert \sqrt{P_i}\Vert_2^2 \Vert D_i\Vert_2^2}{\vartheta_i}.$$
	Additionally, with $\mathbf L_i = [\mathbf J_i^\top~ \mathbf G_i^\top]^\top$, the gathered data from each subsystem satisfies
	\begin{align*}
		\mathbf X_i^{\mathsf d} \! = \! A_i\mathbf J_i + B_i\mathbf G_i + D_i\mathbf W_i - \Delta_i \! = \! R_i \mathbf L_i + D_i\mathbf W_i - \Delta_i,
	\end{align*}
	from which one can clearly see $\Delta_i = -(\underbrace{\mathbf X_i^{\mathsf d}-D_i\mathbf W_i}_{\tilde{\mathbf X}_i^{\mathsf d}}$ $-R_i \mathbf L_i)$. As per Assumption~\ref{asmp:noise assumption}, one has
	\begin{align*}
		\begin{bmatrix}
			\mathds I_{n_i}\\
			R_i^\top
		\end{bmatrix}^{\!\!\top} \overbrace{\begin{bmatrix}
				{\tilde{\mathbf X}_i^{\mathsf d}\tilde{\mathbf X}_i^{\mathsf d}}^{\!\top}\!\!-\!\Lambda_i\Lambda_i^\top & -\tilde{\mathbf X}_i^{\mathsf d}\mathbf L_i^\top\\
				\star & \mathbf L_i \mathbf L_i^\top
		\end{bmatrix}}^{\mathds Y_i}\begin{bmatrix}
			\mathds I_{n_i}\\
			R_i^\top
		\end{bmatrix}\preceq 0.
	\end{align*}
	Using a similar argument as in the proof of Theorem~\ref{thm:main}, the rest of the proof is straightforward and omitted here.$\hfill\blacksquare$

\begin{remark}\label{rem:rank condition}
	To avoid a structural obstruction from a singular bottom-right block in~\eqref{eq:main con} in Assumption~\ref{asmp:main con}, considering~\eqref{eq:main conditions}, matrix $ \mathbf Q_i  \mathbf Q_i^\top$ should have full rank (\ie, $\mathbf Q_i  \mathbf Q_i^\top \succ 0$), meaning that $\operatorname{rank}(\mathbf Q_i) = M_i + N_i + \sigma_i$. Since $\operatorname{rank}(\mathbf Q_i) \leq \min\{M_i + N_i + \sigma_i, T\}$, one needs to collect \emph{at least} $T \geq M_i + N_i + \sigma_i$ samples to ensure that $\mathbf Q_i \mathbf Q_i^\top$ has full rank. In practice, this may not be enough, and the required amount of data should be adequately large to ensure this rank condition. Two additional points are worth noting in this regard. First, condition~\eqref{eq:main con} can possibly be feasible even if $\mathbf Q_i \mathbf Q_i^\top$ does not have full rank. In this case, however, condition~\eqref{eq:main con} is feasible only if the off-diagonal block in~\eqref{eq:main con} lies pointwise in $x_i$ within $\mathscr{R}(\mathbf Q_i \mathbf Q_i^\top)$, where $\mathscr{R}$ denotes the range of $\mathbf Q_i \mathbf Q_i^\top$. This requirement is nontrivial and, in fact, constitutes a strong additional restriction. This observation motivates our interest in ensuring that the matrix $\mathbf Q_i \mathbf Q_i^\top$ has full rank. Second, in~\eqref{eq:main con 2} of Assumption~\ref{asmp:main con 2}, where $D_i$ is assumed to be known, one can observe that at least $T \geq M_i + N_i$ samples are required from each subsystem, which is significantly fewer than in the case of Theorem~\ref{thm:main}. Requiring fewer samples in this setting is also expected, since additional information about the network structure (\emph{i.e.}, matrix $D_i$) is available.
\end{remark}

\begin{remark}
	The availability of $D_i$ significantly improves the scalability of the approach (regarding the connectivity degree) and enables us to deal with densely interconnected networks (cf. Table~\ref{tab:1}). More precisely, as the connectivity degree increases, {$\mathtt{Card}(\mathbf M_i)$} becomes larger and $\sigma_i = \!\underset{j\in\mathbf M_i}{\sum}\! n_j$ grows accordingly. As can be seen, the dimensions of matrices involved in condition~\eqref{eq:main con} directly rely on $\sigma_i$ since $\mathds H_i(x_i),\mathds P_i,\mathds Z_i\in\R^{(n_i + s_i)\times (n_i + s_i)}$ with $s_i = M_i + N_i +\sigma_i$, while the size of the matrices in condition~\eqref{eq:main con 2} are independent of $\sigma_i$. As an example, consider a dense infinite network of scalar subsystems, \ie, $n_i = 1$, with $M_i = 1$ and $N_i = 3$, where {$\mathtt{Card}(\mathbf M_i) = 995$}. For this example, if $D_i$ is unknown, we need to satisfy condition~\eqref{eq:main con}, which is a matrix of dimensions $1000\times 1000$, while, if $D_i$ is available, one can satisfy condition~\eqref{eq:main con 2} instead, reducing the dimensions to $5 \times 5$, which is more scalable as the connectivity degree no longer matters. It is worth noting that while the upper bound on $\|D_i\|_2$ would still suffice to calculate $\rho_i$, the matrix $D_i$ explicitly appears in condition~\eqref{eq:main con 2} through $\mathds Y_i$, necessitating it to be known. Despite the above discussion, we emphasize that Theorem~\ref{thm:main} remains broadly applicable and can be employed for several well-known network topologies, including, but not limited to, cascade and ring topologies.
\end{remark}

The data-driven construction of ISS Lyapunov functions for the individual subsystems, together with their corresponding input-to-state stabilizing controllers, either via Theorem~\ref{thm:main} under Assumptions~\ref{asmp:noise assumption}--\ref{asmp:main con} or Theorem~\ref{thm:main 2} under Assumptions~\ref{asmp:noise assumption} and~\ref{asmp:main con 2}, enables the synthesis of a CLF and a controller for the entire infinite network that guarantee UGAS. To this end, the next section presents a compositional approach based on small-gain reasoning, which systematically combines the ISS Lyapunov functions of the subsystems, purely designed based on data, into a CLF for the infinite network.

\section{Composed Control Lyapunov Function for Infinite Network}\label{sec:compose}

In this section, we derive a sum-type compositional condition under which the ISS Lyapunov functions obtained from data for individual subsystems yield a CLF for the infinite network defined in Definition~\ref{def:inf net}. To do so, we leverage the interconnection constraint~\eqref{eq:int constraint} and propose a small-gain condition in terms of a linear gain operator~\citep{kawan2020lyapunov}.

According to Theorem~\ref{thm:main} {under Assumptions~\ref{asmp:noise assumption}--\ref{asmp:main con}}, and Theorem~\ref{thm:main 2} {under Assumptions~\ref{asmp:noise assumption} and \ref{asmp:main con 2}}, each subsystem $\Sigma_i$ admits an input-to-state stabilizing controller $u_i$ and an ISS
Lyapunov function $\Vs,$ satisfying conditions~\eqref{eq:ISS-con1} and \eqref{eq:ISS-con2}, with $\underline{\alpha}_i, \overline{\alpha}_i,\kappa_i \in \R^+, \rho_i \in \R^+_0$. We define infinite matrices $\mathcal K\Let \Diag{i}{\kappa_i}$ and $\Theta \Let (\theta_{ij})_{i,j\in\N^+},$ with $\theta_{ij} = \frac{\rho_i}{\underline{\alpha}_j},$
where $\theta_{ij}=0, j\notin \mathbf M_i$ {(cf.~\eqref{eq:coefficients} in the proof of Theorem~\ref{thm:comp} for the role of these gains)}. We further define
\begin{align}\label{eq:gain operator}
	\Omega\Let \mathcal K^{-1}\Theta = (\varpi_{ij})_{i,j\in\N^+}, \quad \varpi_{ij} \Let \frac{\theta_{ij}}{\kappa_i}.
\end{align}
We assume there exist constants
${\underline\zeta,\overline\zeta},\underline\kappa\in\R^+$ and $\overline\rho\in\R^+$,
such that $\forall i\in\N^+$,
\begin{align}\label{eq:coeffs conditions}
	0<{\underline\zeta} \le \underline{\alpha}_i\leq\overline{\alpha}_i\le {\overline\zeta}<\infty,\quad 
	\underline\kappa \le \kappa_i,\quad \rho_i\le \overline\rho,
\end{align}
and that the matrix $\Theta$ satisfies
\begin{align}
	\|\Theta\|_{1,1} := \sup_{j\in\N^+}\sum_{i\in\N^+}\theta_{ij} < \infty,
\end{align}
which is equivalent to $\Theta\!:\ell^1\to\ell^1$ being bounded. Here, the double index indicates that the operator norm is induced by the $\ell^1$-norm on both the domain and the co-domain of the operator $\Theta$. Then, $\Omega,$ as defined in~\eqref{eq:gain operator}, serves as the gain operator $\Omega\!:\ell^1\to \ell^1$ by
\begin{align}\label{eq:gain operator 1}
	(\Omega x)_i = \sum_{j\in\N^+}\varpi_{ij}x_j, \quad \forall i\in\N^+\!,
\end{align}
which is a bounded linear operator on $\ell^1$.
We further assume that the spectral radius of the gain operator denoted by $r(\Omega)$ satisfies
\begin{align}\label{eq:con spectral}
	r(\Omega)< 1.
\end{align}

Under the above-mentioned assumptions, we present the following theorem, which provides the conditions required to compose the data-driven ISS Lyapunov functions of the subsystems into a CLF for the infinite network~\citep{kawan2020lyapunov}.

\begin{theorem}\label{thm:comp}
	Consider an infinite network $\Siginf$ as described in Definition~\ref{def:inf net}. Given an input-to-state stabilizing controller $u_i$ and an
	ISS Lyapunov function $\mathcal V_i$  for all $i\in\N^+,$ derived from data using Theorem~\ref{thm:main} {under Assumptions~\ref{asmp:noise assumption}--\ref{asmp:main con}}, or Theorem~\ref{thm:main 2} {under Assumptions~\ref{asmp:noise assumption} and \ref{asmp:main con 2}},
	let conditions~\eqref{eq:coeffs conditions}--\eqref{eq:con spectral} hold. Then for every $\varepsilon>0,$ there exists a vector
	$\mu=(\mu_i)_{i\in\N^+}\in\ell^\infty,$ where
	\begin{align*}
		&\underline\mu \le \mu_i \le \overline\mu,\quad \forall i\in\N^+,\\
		\text{with }&\underline{\mu} \Let \inf_i(\mu_i), \quad \overline{\mu}\Let \sup_i(\mu_i),
	\end{align*}
	and a constant $\kappa_\infty>0,$ where
	\begin{align*}
		\kappa_\infty \ge (1-r(\Omega))\,\underline\kappa-\varepsilon,
	\end{align*}
	with $\underline\kappa\Let \inf_i(\kappa_i)$. Moreover,
	\begin{align}\label{eq:V composed}
		\V := \sum_{i\in\N^+}\mu_i\Vs,\qquad x=(x_i)_{i\in\N^+}\in X,
	\end{align}
	is a CLF for the infinite network in the sense of Theorem~\ref{thm:CLF} under the decentralized controller
	$u:=(u_i)_{i\in\N^+}\in U$, with
	\begin{align}\label{eq:coeff CLF}
		\underline\alpha:=\underline\mu\,{\underline\zeta},\qquad
		\overline\alpha:=\overline\mu\,{\overline\zeta},\qquad
		\kappa:=\kappa_\infty,
	\end{align}
	where ${\underline\zeta}\Let \inf_i(\underline{\alpha}_i)$,  ${\overline\zeta}\Let \sup_i(\overline{\alpha}_i)$, $\overline{\rho}\Let \sup_i(\rho_i)$, thereby rendering the infinite network $\Sigma_K$ UGAS.
\end{theorem}

\textbf{Proof.}
	We begin by showing that $\V$ as in~\eqref{eq:V composed} satisfies condition~\eqref{eq:con1-network}. From condition~\eqref{eq:ISS-con1}, we have
	\begin{align*}
		|x|_2^2 =\!\! \sum_{i\in\N^+}\!|x_i|_2^2\leq \!\!\sum_{i\in\N^+}\!\frac{1}{\underline\alpha_i}\!\Vs\!\leq\! \eta \!\!\sum_{i\in\N^+} \!\!\Vs\leq \eta \frac{1}{\underline\mu}\V,
	\end{align*}
	where $\eta = \sup_i(\frac{1}{\underline\alpha_i})$. Therefore, we have
	\begin{align*}
		\underline\mu\frac{1}{\eta}|x|_2^2 = \underline\mu\,{\underline\zeta}|x|_2^2=\underline\alpha|x|_2^2\leq\V,
	\end{align*}
	when we choose ${\underline\zeta}=\frac{1}{\eta}=\inf_i(\underline\alpha_i)$ and $\underline\alpha = \underline\mu\,{\underline\zeta}$. Moreover, according to~\eqref{eq:V composed} and the right-hand side of condition~\eqref{eq:ISS-con1}, we have
	\begin{align*}
		\V &= \sum_{i\in\N^+}\mu_i\Vs\leq\overline{\mu}\sum_{i\in\N^+}\Vs\leq\overline{\mu}\sum_{i\in\N^+}\overline{\alpha}_i|x_i|_2^2\\
		&\leq \overline{\mu}\,{\overline\zeta}\sum_{i\in\N^+}|x_i|_2^2 = \overline\alpha |x|_2^2,
	\end{align*}
	with ${\overline\zeta} = \sup_i(\overline{\alpha}_i)$ and $\overline\alpha = \overline{\mu}\,{\overline\zeta}$, resulting in the satisfaction of condition~\eqref{eq:con1-network}.
	
	\begin{algorithm}[t!]
		\caption{Data-driven design of CLF and {UGA-stabilizing} controller}\label{Alg:1}
		\begin{center}
			\begin{algorithmic}[1]
				\REQUIRE 
				Dictionaries $\mathcal F_i(x_i)$ and $\mathcal G_i(x_i)$, noise bound $\Lambda_i\Lambda_i^\top $
				\STATE Apply arbitrary admissible inputs $u_i$ to each subsystem initialized at $x_i(t_0)$
				\STATE Execute the infinite network over $[t_0, t_0 + (T-1)\tau]$
				\STATE Gather data $\mathbf U_i$, $\mathbf W_i$, $\mathbf X_i$, and $\mathbf X_i^{\mathsf d}$ according to \eqref{eq:data}
				\FOR{$i\in\N^+$}\label{line 1}
				\STATE Extract data $\mathbf U_i$, $\mathbf W_i$, $\mathbf X_i$, and $\mathbf X_i^{\mathsf d}$
				\STATE
				Construct $\mathbf J_i$ and $\mathbf G_i$ as in~\eqref{eq:mon traj}
				\STATE Design $\mathbf K_i(x_i)$ and $\Phi_i \Let P_i^{-1}$ using \textsf{SOSTOOLS}, for fixed {$\kappa_i,\vartheta_i\in\R^+$}
				\IF{$D_i$ is unknown (with given $\varkappa_i$)}
				\STATE  by satisfying condition~\eqref{eq:main con}
				\ELSE
				\STATE by satisfying condition~\eqref{eq:main con 2}
				\ENDIF
				\STATE
				Construct ISS Lyapunov function $\Vs = x_i^\top \Phi_i^{-1}x_i = x_i^\top P_ix_i$ and controller $u_i = \mathbf K_i(x_i)\Phi_i^{-1}x_i = \mathbf K_i(x_i)P_ix_i $
				\STATE Compute $ \rho_i =  \frac{\Vert\sqrt{P_i}\Vert_2^2\varkappa_i^2}{\vartheta_i}$ (for unknown $D_i$) or $ \rho_i =  \frac{\Vert\sqrt{P_i}\Vert_2^2\Vert D_i \Vert_2^2}{\vartheta_i}$ (for known $D_i$)
				\ENDFOR\label{line 12}
				{\IF{conditions~\eqref{eq:coeffs conditions}--\eqref{eq:con spectral} are met}
					\STATE
					$\V \Let \sum_{i\in\N^+} \!\mu_i\Vs$ is a CLF for the infinite network $\Sigma$, and controller $u = (u_i)_{i\in\N^+}$ renders the network $\Sigma_K$ UGAS
					\ELSE
					\STATE
					Repeat Steps \ref{line 1}--\ref{line 12} with more collected samples $T$ or different parameters ${\kappa_i},\vartheta_i$
					\ENDIF}
				\ENSURE CLF $\V= \sum_{i\in\N^+} \mu_ix_i^\top P_ix_i$, controller $u = (u_i)_{i\in\N^+}$ with $u_i \!=\! \mathbf K_i(x_i)P_ix_i,$ satisfaction of UGAS property across the infinite network $\Sigma_K$
			\end{algorithmic}
		\end{center}
	\end{algorithm}
	
	\noindent We proceed with demonstrating the satisfaction of condition~\eqref{eq:con2-network}, utilizing the definition of the composed CLF, as
	{\begin{align*}
			\Vp &= \sum_{i\in\N^+}\!\mu_i\dot{\mathcal V}_i(x_i),
		\end{align*}
		with
		\begin{align*}
			\dot{\mathcal V}_i(x_i)  &\coloneq \limsup_{s \to 0} \tfrac{1}{s}\big(\mathcal V_i(\phi_{u_i}(s,x_i,w_i)) - \mathcal V_i(x_i)\big),
		\end{align*}
		which is similar to~\eqref{eq:dini V} but is defined for each subsystem endowed with solution $\phi_{u_i}(s,x_i,w_i)$. Using~\eqref{eq:ISS-con2} as a computationally simpler style of this formulation, one can obtain}
	\begin{align}
		{\Vp}&\overset{\eqref{eq:ISS-con2}}{\leq} \!\sum_{i\in\N^+}\!\mu_i(-\kappa_i\Vs + \rho_i|w_i|_2^2)\notag\\
		& \overset{\eqref{eq:w partition}}{=}\!\sum_{i\in\N^+}\!\mu_i(-\kappa_i\Vs + \!\sum_{j\in\mathbf M_i}\!\rho_i|w_{ij}|_2^2)\notag\\
		& \overset{\eqref{eq:int constraint}}{=} \!\sum_{i\in\N^+}\!\mu_i(-\kappa_i\Vs + \!\sum_{j\in\mathbf M_i}\!\rho_i|x_j|_2^2)\notag\\
		&\overset{\eqref{eq:ISS-con1}}{\leq} \!\!\sum_{i\in\N^+}\!\mu_i(-\kappa_i\Vs + \!\sum_{j\in\mathbf M_i}\!\frac{\rho_i}{\underline\alpha_j}\mathcal V_j(x_j))\notag\\
		&\overset{\hphantom{\eqref{eq:ISS-con1}}}{=}\!\!\sum_{i\in\N^+}\!-\kappa_i\mu_i\Vs + \sum_{i\in\N^+}\!\sum_{j\in\mathbf M_i}\mu_i\theta_{ij}\mathcal V_j(x_j).\label{eq:coefficients}
	\end{align}
	{Since $\theta_{ij} = 0, j\notin\mathbf M_i$ and $i\notin\mathbf M_i$ according to Definition~\ref{def:ct-NPS}, {$\theta_{ii} = 0$}. Then, we have
		\begin{align*}
			{\Vp}&\leq\sum_{i\in\N^+}\!-\kappa_i\mu_i\Vs + \sum_{i\in\N^+}\!\sum_{j\in\N^+}\mu_i\theta_{ij}\mathcal V_j(x_j).
		\end{align*}
		Now, by changing the index of the first summation to $j$, without loss of generality, and re-grouping by $j$, one has}
	\begin{align*}
		{\Vp}&\leq\sum_{j\in\N^+}\!\!\big(-\kappa_j\mu_j +\! \sum_{i\in\N^+}\!\mu_i\theta_{ij}\big)\mathcal V_j(x_j).
	\end{align*}
	{Satisfying conditions~\eqref{eq:coeffs conditions}--\eqref{eq:con spectral}, according to~\citet{kawan2020lyapunov}, yields
		$$
		\sum_{i\in\N^+}\mu_i\theta_{ij}\le (\kappa_j-\kappa_\infty)\mu_j,\qquad \forall j\in\N^+
		$$
		being satisfied. Then, one can obtain}
	\begin{align*}
		{\Vp}&\leq\!\!\!\sum_{j\in\N^+}\! \!\!-\kappa_\infty\mu_j\mathcal V_j(x_j)\!=\!-\kappa_\infty\!\!\!\sum_{j\in\N^+}\! \!\mu_j\mathcal V_j(x_j)=-\kappa\V,
	\end{align*}
	with $\kappa\Let\kappa_\infty$, which concludes the proof.$\hfill\blacksquare$

\begin{table*}[h!]
	\centering
	\caption{An overview of the data-driven results for infinite networks, reporting the per-subsystem runtime in seconds (\texttt{RT}) and memory usage (\texttt{MU}) in megabytes. Here, {$\mathtt{Card}(\mathbf M_i)$}
		denotes the number of subsystems interconnected with the $i$-th subsystem, and $T$ denotes the number of collected samples. For each system, results for the network parameters and computation costs are reported in two rows: the first row is for the case where $D_i$ is unknown, while the second row is for the case where $D_i$ is known.\label{tab:1}}\vspace{0.2cm}
	\resizebox{\textwidth}{!}{\begin{tabular}{@{}cccccccc@{}}
			\toprule
			\multirow{2}{*}[-0.25em]{System} & \multirow{2}{*}[-0.25em]{Dictionary $\mathcal F_i(x_i)$} & \multirow{2}{*}[-0.25em]{Dictionary $\mathcal G_i(x_i)$} & \multicolumn{2}{c}{Network parameters} & \multicolumn{2}{c}{Computation costs}\\
			\cmidrule(lr){4-5} \cmidrule(lr){6-7}
			{} & {} & {} & {$\mathtt{Card}(\mathbf M_i)$} & $T$  & \texttt{RT} (sec) & \texttt{MU} (MB)\\
			\midrule
			\myalign{l}{\multirow{2}{*}{Spacecraft}} & \multirow{2}{*}{$[x_{1_i}~ x_{2_i}~x_{3_i}~x_{1_i}x_{2_i}~x_{2_i}x_{3_i}~x_{1_i}x_{3_i}]^{\!\top}$} & \multirow{2}{*}{$\mathds I_3$} & {$1$} &{$70$}  & {$\approx 1$} & {$0.56$} \\
			{} & {} & {} & {$1800$} & {$50$} & {$15.53$} & {$2.84$}\\
			\midrule
			\myalign{l}{\multirow{2}{*}{Lorenz}}	& \multirow{2}{*}{$[x_{1_i}~ x_{2_i}~x_{3_i}~x_{1_i}x_{3_i}~x_{1_i}x_{2_i}~x_{2_i}x_{3_i}]^{\!\top}$}	&	\multirow{2}{*}{$\mathds I_3$}	&	{$1$}	&	{$25$}	&		{$1.21$} & {$0.56$}	\\
			{} & {} & {} &{$1000$}	&		{$80$}	&	{$2.83$} & {$2.47$}	\\
			\midrule
			\myalign{l}{\multirow{2}{*}{Academic}} & \multirow{2}{*}{$[x_{1_i}~ x_{2_i}~ x_{1_i}x_{2_i}~ x_{1_i}^2~ x_{2_i}^2]^{\!\top}$}  &  \multirow{2}{*}{$[1~x_{1_i}~x_{2_i}]^{\!\top}$}   &{$5$}	&{$21$}		&	{$1.21$}	& {$0.53$}\\
			{} & {} & {} &{$500$} & {$50$} & {$<1$} & {$0.61$}\\
			\bottomrule
	\end{tabular}}
\end{table*}

Algorithm~\ref{Alg:1} outlines a concise summary of the proposed compositional data-driven method for the construction of a CLF together with a {UGA-stabilizing} controller for an unknown infinite network. The approach builds on data-driven ISS Lyapunov functions associated with the individual subsystems to enable UGAS certification of the overall interconnection.

\section{Simulation Results}\label{sec:simulation}
In this section, we apply our data-driven approach to three different infinite networks to demonstrate its applicability. The case studies include two physical examples, spacecraft and a Lorenz-type chaotic system, and an academic example. For each case, we first assume that all subsystem matrices $A_i$, $B_i$, and $D_i$, and consequently $f(x)$ and $g(x)$ of the infinite network, are unknown, and we assess the results of Theorem~\ref{thm:main}. We then consider networks with denser connectivity by assuming $D_i$ is known and evaluate the performance of Theorem~\ref{thm:main 2}. Table~\ref{tab:1} provides a concise yet comprehensive overview of the results, including a comparison across different levels of interconnection sparsity for scenarios where $D_i$ is unknown and where $D_i$ is available.

Across all case studies, the main objective is to synthesize a CLF together with a {UGA-stabilizing} controller for an infinite network with an unknown mathematical model. To this end, following the required steps outlined in Algorithm~\ref{Alg:1}, input–state trajectory data are first collected from each subsystem in accordance with~\eqref{eq:data}. These data are then used to enforce condition~\eqref{eq:main con} when $D_i$ is unknown, or condition~\eqref{eq:main con 2} otherwise, and thereby construct ISS Lyapunov functions and corresponding local controllers for all subsystems, purely based on data. Subsequently, leveraging the compositionality results established in Theorem~\ref{thm:comp}, a network-level CLF and its associated controller are synthesized in a compositional manner from the data-driven ISS Lyapunov functions, thereby guaranteeing UGAS of the whole infinite network. All simulations are conducted in \MATLAB \textsl{R2023b} on a MacBook Pro equipped with an Apple~M2~Pro processor and 16~GB of memory. The following subsections detail each case study and present the corresponding numerical results.

\subsection{Spacecraft System}\label{subsec:sc}
As the first case study, we consider an infinite network of $3$-dimensional rotating rigid spacecraft subsystems~\citep{khalil2002control}, {with $x_i = [x_{1_i}~x_{2_i}~x_{3_i}]^\top$, where the state variables are the angular velocities along the principal axes.
	Each spacecraft evolves
	according to~\eqref{eq:subsystem reform}, where the control input is $u_i = [u_{1_i}~u_{2_i}~ u_{3_i}]^\top$\!, representing the applied torque, $A_i^* = \operatorname{diag}\big(\frac{(J_{2_i} - J_{3_i})}{J_{1_i}},\frac{(J_{3_i} - J_{1_i})}{J_{2_i}},\frac{(J_{1_i} - J_{2_i})}{J_{3_i}}\big)$, and $B_i^* = \operatorname{diag}\big(\frac{1}{J_{1_i}} ,$ $\frac{1}{J_{2_i}},\frac{1}{J_{3_i}}\big)$, with $J_{1_i}$, $J_{2_i}$, and $J_{3_i}$ denoting the principal moments of inertia of the $i$-th spacecraft. The coupling terms capture the influence of the interconnected subsystems on each subsystem through $D_i w_i$, with $w_i$ satisfying $w_i=(w_{ij})_{j\in\mathbf M_i}=(x_j)_{j\in\mathbf M_i}$ as per \eqref{eq:int constraint}. Hence, considering $D_i$ as defined in~\eqref{eq:D partition} with
	$$
	D_{ij} = -10^{-4}\times\begin{bmatrix}
		0 & 0 & 1\\
		1 & 0 & 0\\
		0 & 1 & 0
	\end{bmatrix}\!\!,
	$$
	we have
	$$
	D_i w_i = -10^{-4}\times\begin{bmatrix}
		\sum_{j\in\mathbf M_i}x_{3_j}\\[2mm]
		\sum_{j\in\mathbf M_i}x_{1_j}\\[2mm]
		\sum_{j\in\mathbf M_i}x_{2_j}
	\end{bmatrix}\!\!.
	$$
	In addition, $\mathcal F_i^*(x_i) = [x_{2_i}x_{3_i}~x_{1_i}x_{3_i}~x_{1_i}x_{2_i}]^\top$ and $\mathcal G_i^*(x_i) = \mathds I_3$, which, together with other matrices describing the dynamics, are all \emph{unknown}. Nevertheless, we have}
{$\mathcal F_i(x_i) = [x_{1_i}~x_{2_i}~x_{3_i}~x_{1_i}x_{2_i}~x_{2_i}x_{3_i}~x_{1_i}x_{3_i}]^\top,$} based on which we can choose the transformation matrix as {$\Psi_i(x_i) = [\mathds I_3~\operatorname{diag}(x_{2_i},x_{3_i},x_{1_i})]^\top$}, and $\mathcal G_i(x_i) = \mathds I_3$.  Furthermore, we consider the elements of $\Delta_i$ take values in the interval {$[-0.01, 0.01]$}.
\begin{figure*}[t!]
	\centering
	\subfloat[Evolution of open-loop subsystems\label{fig:sc ol}]{
		\includegraphics[width=0.33\linewidth]{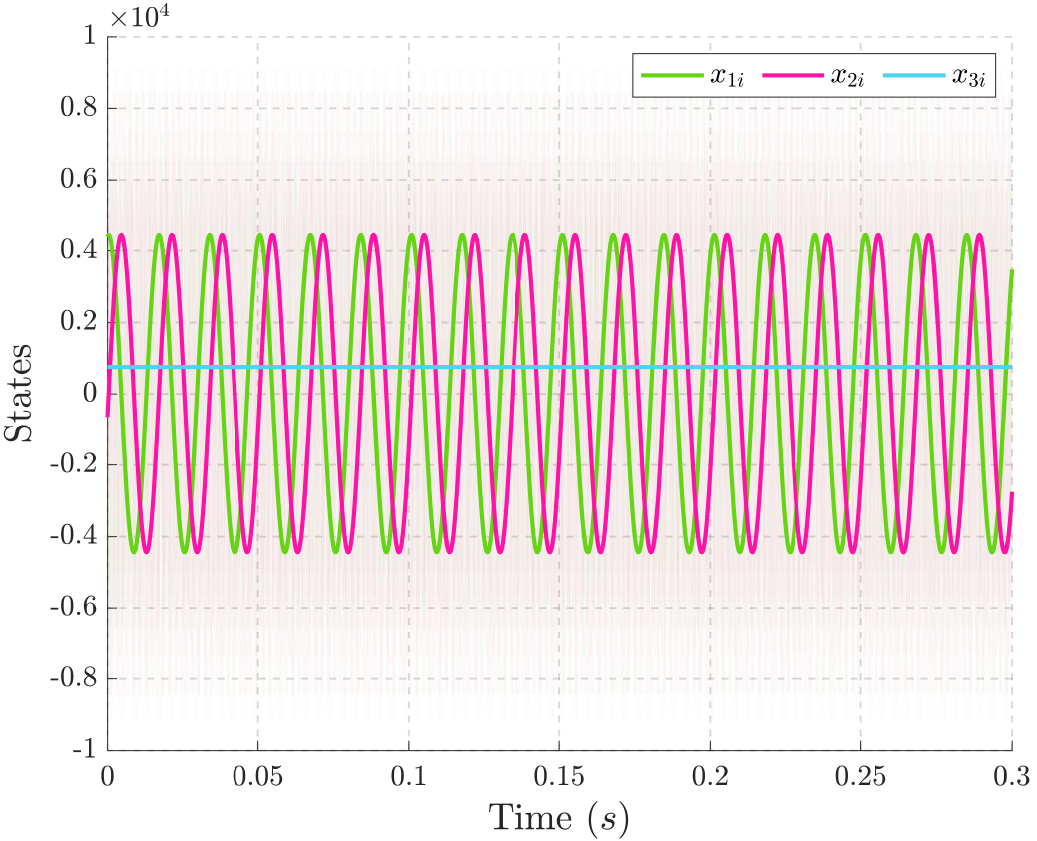}}
	\subfloat[\centering Evolution of closed-loop subsystems with unknown $D_i$\label{fig:sc unknown}]{
		\includegraphics[width=0.33\linewidth]{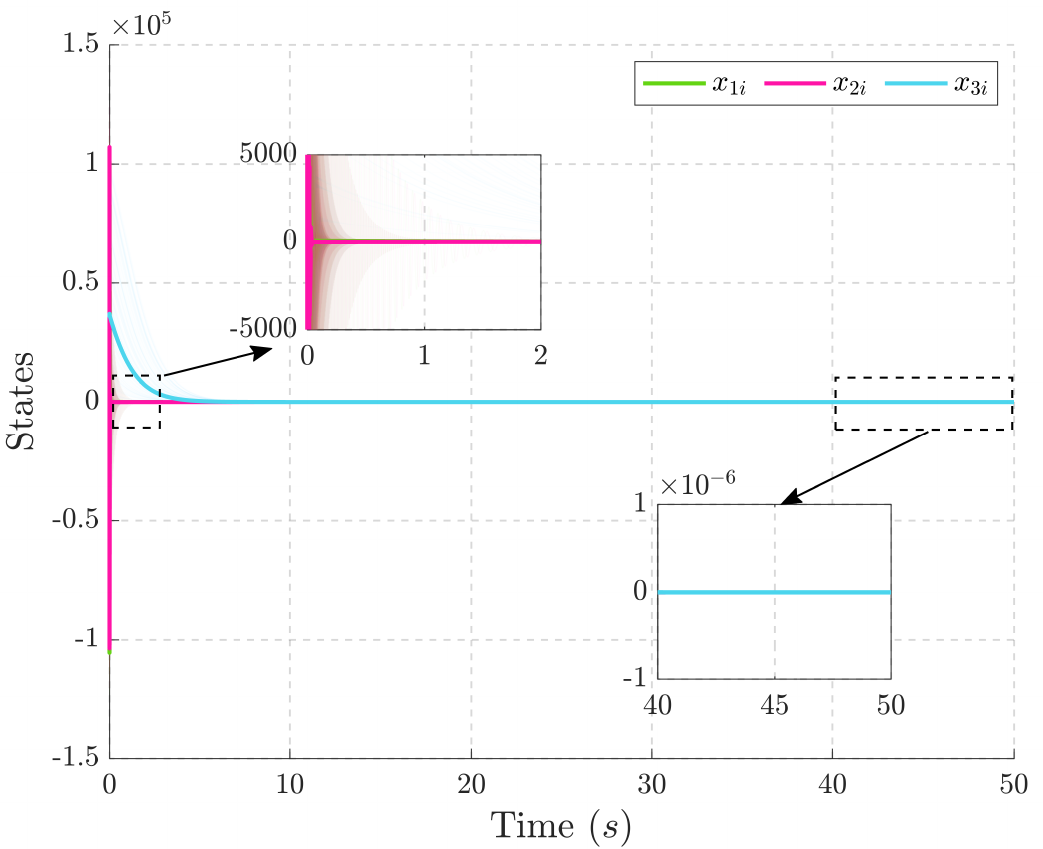}}
	\subfloat[\centering Evolution of closed-loop subsystems with known $D_i$\label{fig:sc known}]{
		\includegraphics[width=0.33\linewidth]{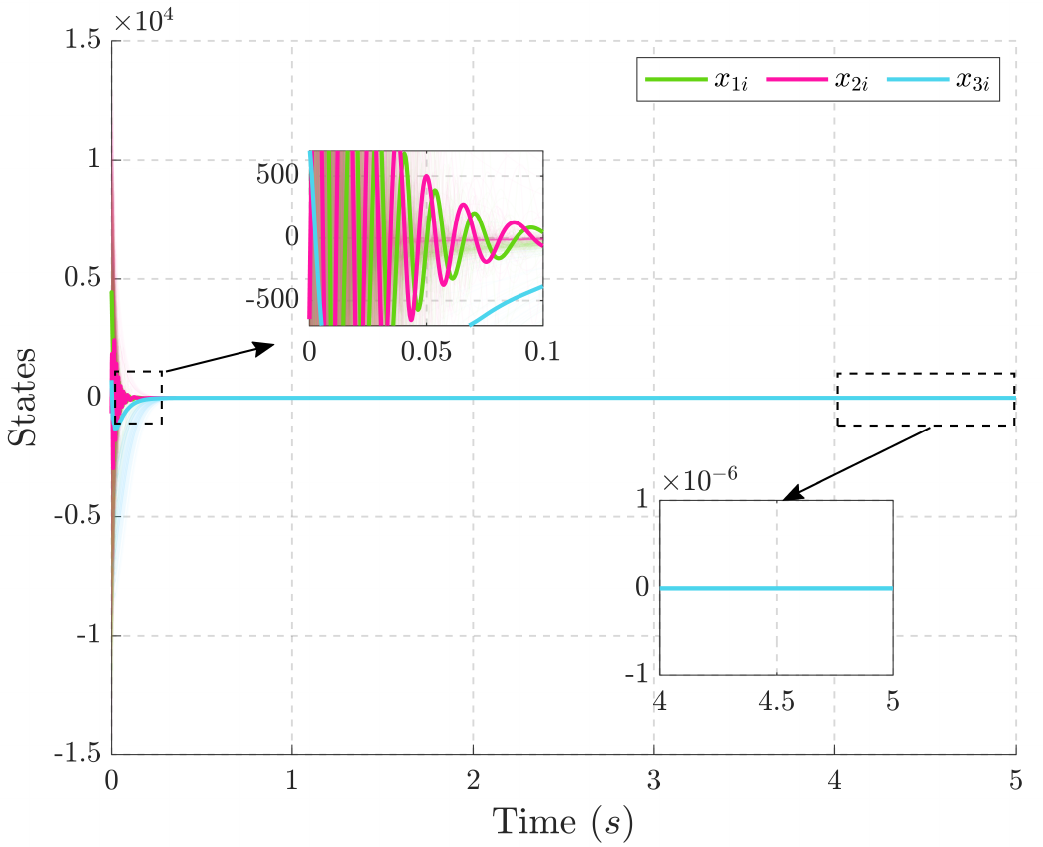}}
	\caption{\textbf{Spacecraft system:}
		Evolution of the subsystems under the synthesized controller~\eqref{eq:u SC unknown} (for unknown $D_i$) or \eqref{eq:u SC known} (for known $D_i$), with the states of a representative subsystem highlighted in bold and the remaining subsystems shown in a faded style. Fig.~(b) corresponds to the case where $D_i$ is unknown, Fig.~(c) considers the case where $D_i$ is known, while Fig.~(a) exhibits subsystems' open-loop behavior.\label{fig:sc}}
\end{figure*}

\textbf{Known Upper Bound on $\pmb{\|D_i\|_2}$.} 
Although $D_i$ is assumed to be unknown, {$\|D_i\|_2\leq 0.05$} is given. The states of each subsystem are affected by $j\in\mathbf M_i =\{i-1\},$ for all $i\geq2$, where $\mathtt{Card}(\mathbf M_i)= 1$, forming a cascade or line topology.
We collect {$T = 70$} samples with a sampling period of {$\tau = 0.1$}, and according to Remark~\ref{rem:noise}, we have {$\Lambda_i\Lambda_i^\top = 0.021\mathds I_3$}. We set {$\kappa_i = 0.1$ and $\vartheta_i = 1$}, and given the samples and dictionaries for the $i$-th subsystem, we follow the steps outlined in Algorithm~\ref{Alg:1}, and obtain
\begin{align}
	P_i =&~ 10^6\times \begin{bmatrix}
		1.5232  &  0.1830 &  -0.2349\\
		\star  &  1.0210  &  0.0435\\
		\star &  \star  &  1.9255\\
	\end{bmatrix}\!\!,\notag\\
	\Vs =& ~1523155.6413\,x_{1_i}^2 + 366098.4124\,x_{1_i}x_{2_i}\notag\\
	&- 469884.9578\,x_{1_i}x_{3_i} + 1020988.6223x_{2_i}^2 \notag\\
	&+ 87094.6687\,x_{2_i}x_{3_i} + 1925500.0465\,x_{3_i}^2, \notag\\
	u_{1_i} =& -5.2653\times 10^{-10}x_{1_i}^2 - 0.0097219\,x_{1_i}x_{2_i} \notag\\
	&- 1.7263\,x_{1_i}x_{3_i} + 6.0786\times 10^{-5}x_{2_i}^2\notag \\
	&- 0.45358\,x_{2_i}x_{3_i} - 0.029396\,x_{3_i}^2 - 456.7931\,x_{1_i}\notag\\
	&- 76.8262\,x_{2_i} + 63.3713\,x_{3_i},\notag\\
	u_{2_i} = & -8.2599\times 10^{-11}x_{1_i}^2 + 0.016265\,x_{1_i}x_{2_i}\notag \\
	&- 0.045597\,x_{1_i}x_{3_i} + 4.3491\times 10^{-5}x_{2_i}^2 \notag\\
	&- 0.22552\,x_{2_i}x_{3_i} - 0.14856\,x_{3_i}^2 - 216.9743\,x_{1_i} \notag\\
	&- 271.3122\,x_{2_i} - 23.2363\,x_{3_i},\notag\\
	u_{3_i} = & -1.1071\times 10^{-10}x_{1_i}^2 - 0.0036229\,x_{1_i}x_{2_i}\notag\\
	&- 0.67897\,x_{1_i}x_{3_i} - 0.00012082\,x_{2_i}^2  \notag\\
	&+ 0.24018\,x_{2_i}x_{3_i}+ 0.0030564\,x_{3_i}^2 + 25.329\,x_{1_i}\notag\\
	& - 36.8423\,x_{2_i} - 287.6872\,x_{3_i}.\label{eq:u SC unknown}
\end{align}
Furthermore, we attain {$\underline{\alpha}_i = 9.4701\times 10^5,$ $\overline{\alpha}_i = 2.0351 \times 10^6,$ and $\rho_i = 5.0876\times 10^3$}. {Then, we need to verify whether condition~\eqref{eq:con spectral} holds with the gains obtained using data.} We construct the infinite matrix $\Omega=(\varpi_{ij})_{i,j\in\N^+}$ as in~\eqref{eq:gain operator}, which is in the following form:
\begin{align*}
	\Omega = 
	\begin{bmatrix}
		0 & 0 & 0 & 0 & 0  & \ldots \\
		0.0537 & 0 & 0 & 0 & 0  & \ldots \\
		0 & 0.0537 & 0 & 0 & 0 & \ldots \\
		0 & 0 & 0.0537 & 0 & 0 &  \ldots \\
		\vdots & \ddots & \ddots & \ddots & \ddots  & \ddots
	\end{bmatrix}\!\!.
\end{align*}
To verify the small-gain condition~\eqref{eq:con spectral}, recall that
\begin{align}\label{eq:SGC}
	r(\Omega)\leq \|\Omega\|_{1,1} = \sup_{j\in\N^+}\sum_{i\in\N^+}\varpi_{ij}.
\end{align}
Owing to the form of $\Omega$ and to the fact that all subsystems are identical,
we obtain $r(\Omega) \leq \|\Omega\|_{1,1} \leq 0.0537<1$, resulting in the satisfaction of condition~\eqref{eq:con spectral}. This enables us to construct $\V = \sum_{i\in\N^+}\mu_i\Vs$ as the CLF for the infinite network and $u = (u_i)_{i\in\N^+}$ as its {UGA-stabilizing} controller, ensuring that the infinite network is UGAS.

\textbf{Known Matrix $\pmb{D_i}$.} Assuming $D_i$ is known, we consider a dense infinite network with $\mathtt{Card}(\mathbf M_i) =1800$, where each subsystem is affected by $1800$ subsequent subsystems, {that is, $\mathbf M_i=\{i+1,\ldots,i+1800\}$.} We obtain {$T = 50$} samples with a sampling period of {$\tau = 0.1$}, resulting in {$\Lambda_i\Lambda_i^\top = 0.015\mathds I_3$ as per Remark~\ref{rem:noise}}. By setting {$\kappa_i = 0.1$} and {$\vartheta_i = 5.5$}, and following the steps proposed in Algorithm~\ref{Alg:1}, we design
\begin{align}
	P_i = &10^5\times\begin{bmatrix}
		2.8221  &  0.3032 &  -0.3600\\
		\star &  3.0950 &  -0.7946\\
		\star & \star &   6.0998\\
	\end{bmatrix}\!\!,\notag\\
	\Vs = &~282214.8765\,x_{1_i}^2 + 60631.9048\,x_{1_i}x_{2_i}\notag \\
	&- 72001.835\,x_{1_i}x_{3_i} + 309503.6109\,x_{2_i}^2 \notag\\
	&- 158910.8974\,x_{2_i}x_{3_i} + 609976.6158\,x_{3_i}^2, \notag\\
	u_{1_i} =& -0.44592\,x_{1_i}^2 - 2.7864\,x_{1_i}x_{2_i} + 7.4505\,x_{1_i}x_{3_i} \notag\\
	&- 14.11\,x_{2_i}^2 + 10.4307\,x_{2_i}x_{3_i} + 0.232\,x_{3_i}^2 \notag\\
	&- 4884.1696\,x_{1_i} - 130.7843\,x_{2_i} \!+\! 140.8832\,x_{3_i},\notag\\
	u_{2_i} = & -0.94461\,x_{1_i}^2 \!+\! 12.2549\,x_{1_i}x_{2_i} \!+\! 0.50344\,x_{1_i}x_{3_i} \notag\\
	&+ 1.0336\,x_{2_i}^2 - 0.30186\,x_{2_i}x_{3_i} - 2.5638\,x_{3_i}^2\notag\\
	&- 63.8954\,x_{1_i} - 4376.8145\,x_{2_i} + 4.4998\,x_{3_i},\notag\\
	u_{3_i} = & -4.0831\,x_{1_i}^2 - 5.4064\,x_{1_i}x_{2_i} + 1.8243\,x_{1_i}x_{3_i} \notag\\
	&- 1.195\,x_{2_i}^2 + 2.5113\,x_{2_i}x_{3_i} - 0.58727\,x_{3_i}^2 \notag\\
	&+ 39.996\,x_{1_i} - 37.1488\,x_{2_i} - 6081.5056\,x_{3_i}.\label{eq:u SC known}
\end{align}
Moreover, we compute {$\underline{\alpha}_i = 2.6191\times 10^5,$ $\overline{\alpha}_i = 6.3477\times 10^5,$ and $\rho_i = 2.0774,$ leading to the fulfillment of condition~\eqref{eq:con spectral} with {$r(\Omega) \leq \|\Omega\|_{1,1} \leq 0.1428<1$ using~\eqref{eq:SGC}}, where
	\begin{align*}
		\Omega = 10^{-5}\!\times\!
		\begin{bNiceMatrix}[first-row]
			&\Hdotsfor{3}[line-style=mybrace, shorten=0pt]^{\ensuremath{j\in\{2,\ldots,1801\}}}\\
			0 & 7.9318 & \ldots & 7.9318 & 0  & \ldots \\
			0 & 0 & 7.9318 & \ldots & 7.9318  & \ldots \\
			0 & 0 & 0 & 7.9318 & \ldots & \ldots \\
			0 & 0 & 0& 0 & 7.9318 &  \ldots \\
			\vdots & \ddots & \ddots & \ddots & \ddots  & \ddots
		\end{bNiceMatrix}\!.
	\end{align*}
	\begin{figure*}[t!]
		\centering
		\subfloat[Evolution of open-loop subsystems\label{fig:lor ol}]{
			\includegraphics[width=0.33\linewidth]{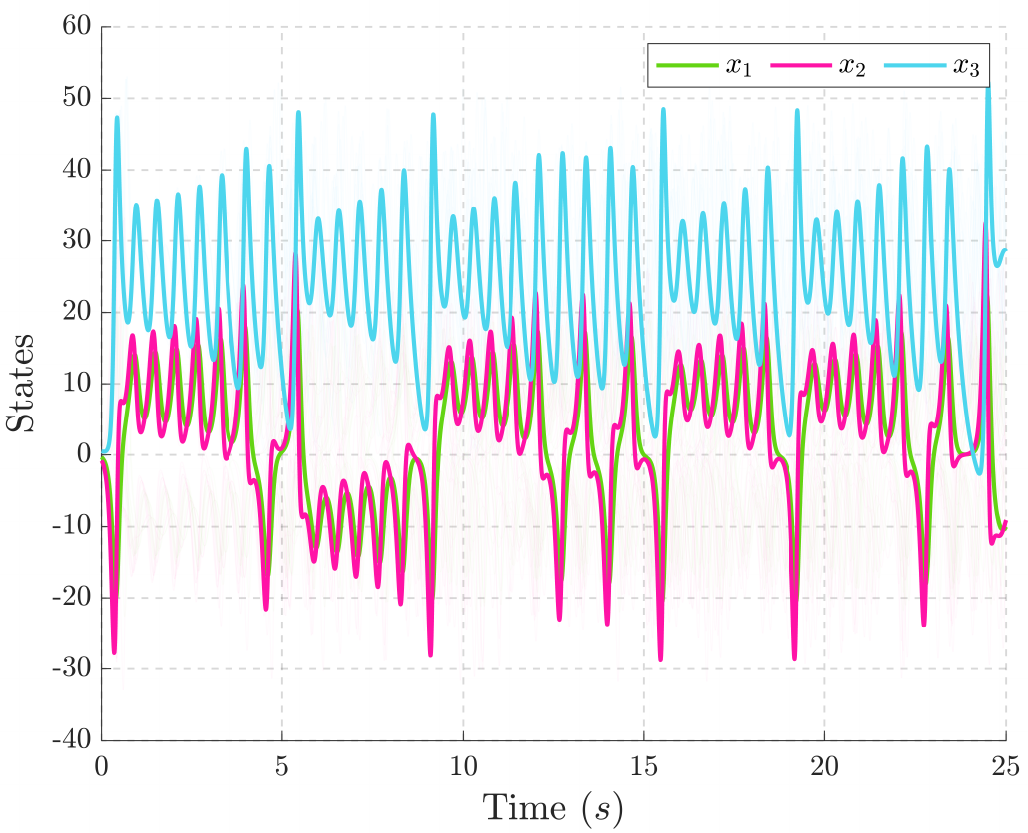}}
		\subfloat[\centering Evolution of closed-loop subsystems with unknown $D_i$\label{fig:lor unknown}]{
			\includegraphics[width=0.33\linewidth]{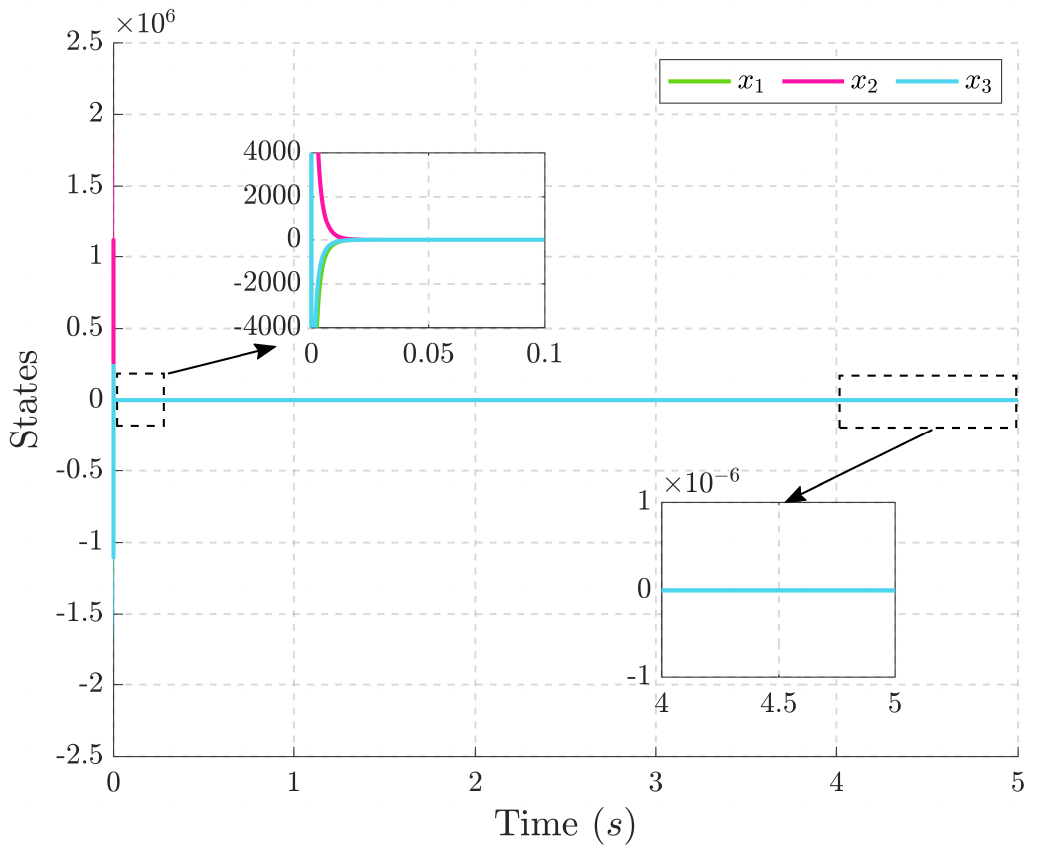}}
		\subfloat[\centering Evolution of closed-loop subsystems with known $D_i$\label{fig:lor known}]{
			\includegraphics[width=0.33\linewidth]{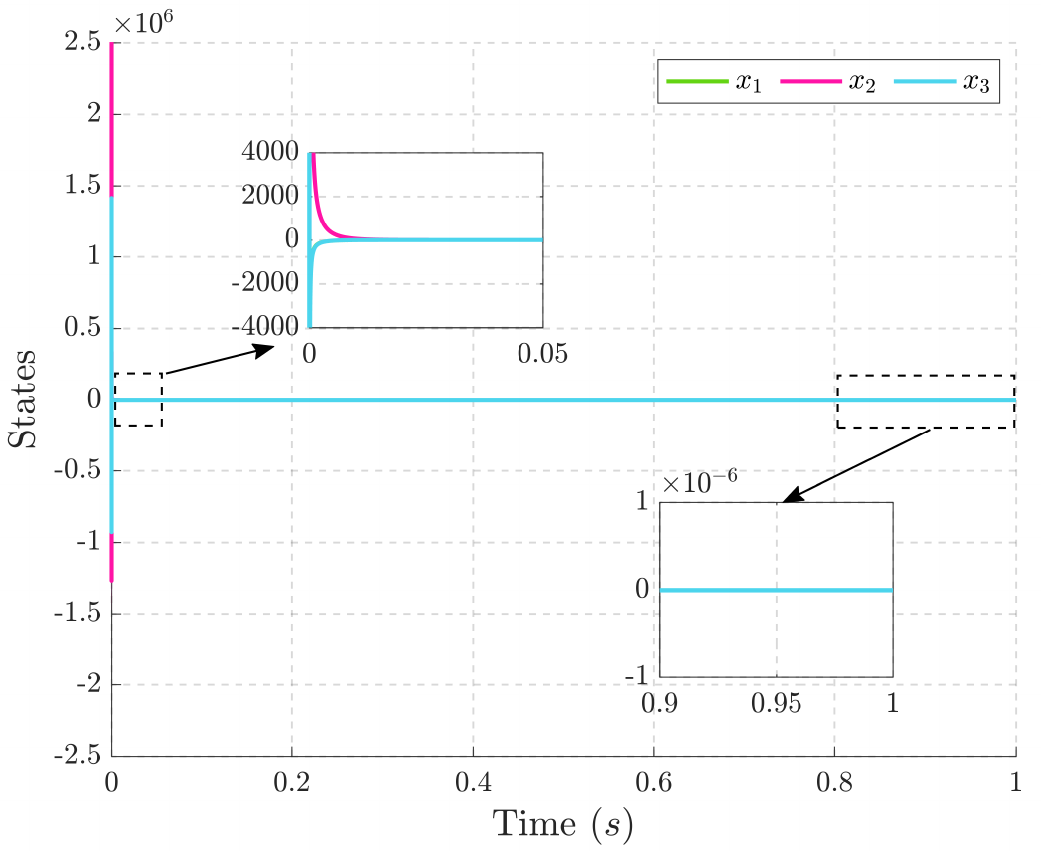}}
		\caption{\textbf{Lorenz system:}
			Illustration of the subsystems’ evolution under the synthesized controller~\eqref{eq:u Lorenz unknown} (for unknown $D_i$) or \eqref{eq:u Lorenz known} (for known $D_i$), with the states of a representative subsystem emphasized in bold and the remaining subsystems shown in a faded style. Fig.~(b) corresponds to the case of unknown $D_i$, Fig.~(c) considers the case where $D_i$ is known, while the open-loop behavior of the subsystems is depicted in Fig.~(a).\label{fig:lor}}
	\end{figure*}	
	Accordingly, we construct the CLF for the infinite network $\V = \sum_{i\in\N^+}\mu_i\Vs$  and its {UGA-stabilizing} controller $u = (u_i)_{i\in\N^+}$ , guaranteeing the infinite network is UGAS. The evolution of some arbitrary subsystems, starting from large initial conditions on the order of $10^4$, under the designed controllers~\eqref{eq:u SC unknown} and \eqref{eq:u SC known} is demonstrated in Fig.~\ref{fig:sc}, ensuring that their states converge to the origin asymptotically.}

\subsection{Lorenz System}
As the second case study, we examine an infinite network of Lorenz systems~\citep{strogatz2024nonlinear}, which are well-known for exhibiting chaotic dynamics, \emph{i.e.}, extreme sensitivity to variations in initial conditions. {Networks of Lorenz systems arise in a wide range of applications, including engineering design of robust systems that can tolerate and exploit chaotic behavior~\citep{sprott2010elegant}, secure communications through chaos-based encryption schemes~\citep{wang2009chaotic}, and neuroscience for modeling complex brain activity patterns~\citep{strogatz2024nonlinear}.}

{Each Lorenz subsystem is described by~\eqref{eq:subsystem reform}, with $x_i = [x_{1_i}~x_{2_i}~x_{3_i}]^\top$\!, $u_i = [u_{1_i}~u_{2_i}~u_{3_i}]^\top$\!,
	\begin{align*}
		A_i^* &=
		\begin{bmatrix}
			-10 & 10 & 0 & 0 & 0\\
			28 & -1 & 0 & -1 & 0\\
			0 & 0 & -\frac{8}{3} & 0 & 1
		\end{bmatrix}\!\!, \text{ and }
		B_i^* =
		\mathds I_3.
	\end{align*}
	The effect of other subsystems is captured by $D_iw_i$, with $w_i=(w_{ij})_{j\in\mathbf M_i}=(x_j)_{j\in\mathbf M_i}$ as in~\eqref{eq:int constraint} and $D_i=(D_{ij})_{j\in\mathbf M_i}$ as in~\eqref{eq:D partition}, where
	\begin{align*}
		D_{ij} &= 10^{-3}\times\begin{bmatrix}
			1 & 0 & 0\\
			0 & 0 & 0 \\
			0 & 0 & -1
		\end{bmatrix}\!\!,
	\end{align*}
	which yields
	\begin{align*}
		D_i w_i=10^{-3}\times
		\begin{bmatrix}
			\sum_{j\in\mathbf M_i}x_{1_j}\\[2mm]
			0\\[1mm]
			-\sum_{j\in\mathbf M_i}x_{3_j}
		\end{bmatrix}\!\!.
	\end{align*}
	We assume that these matrices, along with}
$\mathcal F^*_i(x_i) = [x_{1_i}~x_{2_i}~ x_{3_i}~ x_{1_i}x_{3_i}~ x_{1_i}x_{2_i}]^\top\!,$ and $\mathcal G_i^*(x_i) = \mathds I_3,$ are all unknown, while we have dictionaries $\mathcal F_i(x_i) = [x_{1_i}~ x_{2_i}~ x_{3_i}~ x_{1_i}x_{3_i}~ x_{1_i}x_{2_i}~ x_{2_i}x_{3_i}]^\top$ and $\mathcal G_i(x_i) = \mathds I_3$. We therefore consider the transformation matrix as $\Psi_i(x_i) = [\mathds I_3~ \operatorname{diag}(x_{3_i},x_{1_i},x_{2_i})]^\top$.
Furthermore, we assume that the elements of $\Delta_i$ take values in the interval {$[-0.001, 0.001]$}.

\begin{figure*}[t!]
	\centering
	\subfloat[Evolution of open-loop subsystems\label{fig:acad ol}]{
		\includegraphics[width=0.33\linewidth]{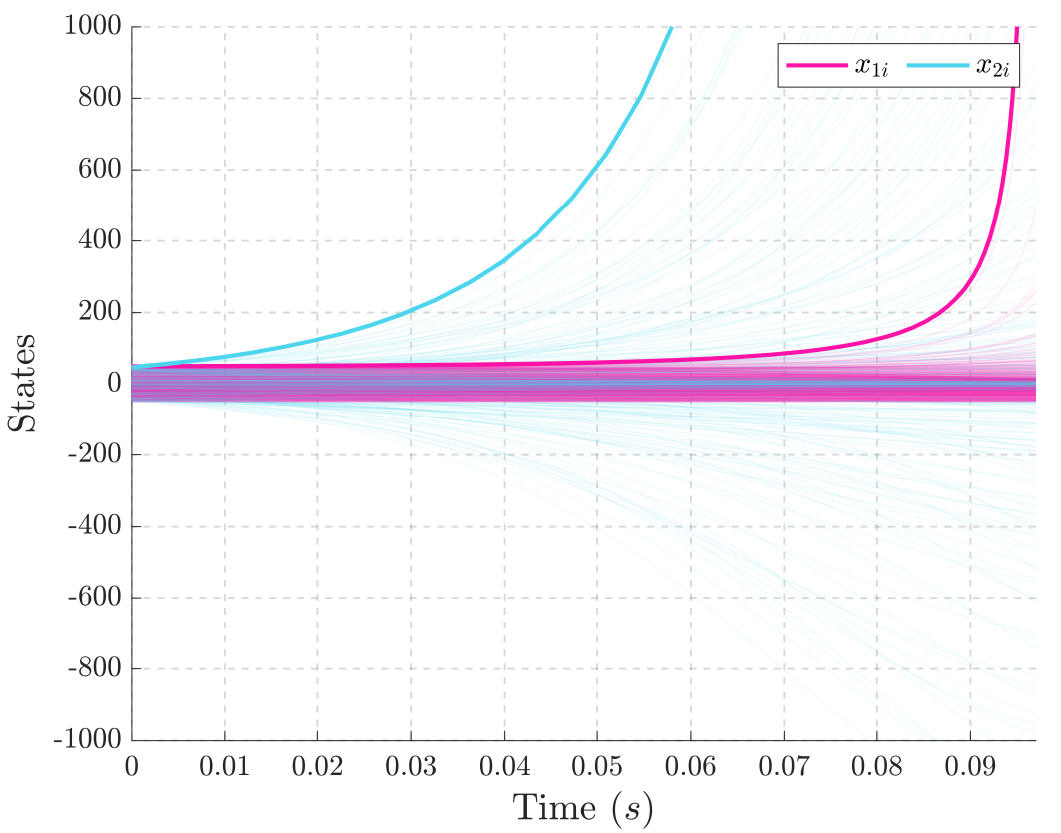}}
	\subfloat[\centering Evolution of closed-loop subsystems with unknown $D_i$\label{fig:acad unknown}]{
		\includegraphics[width=0.33\linewidth]{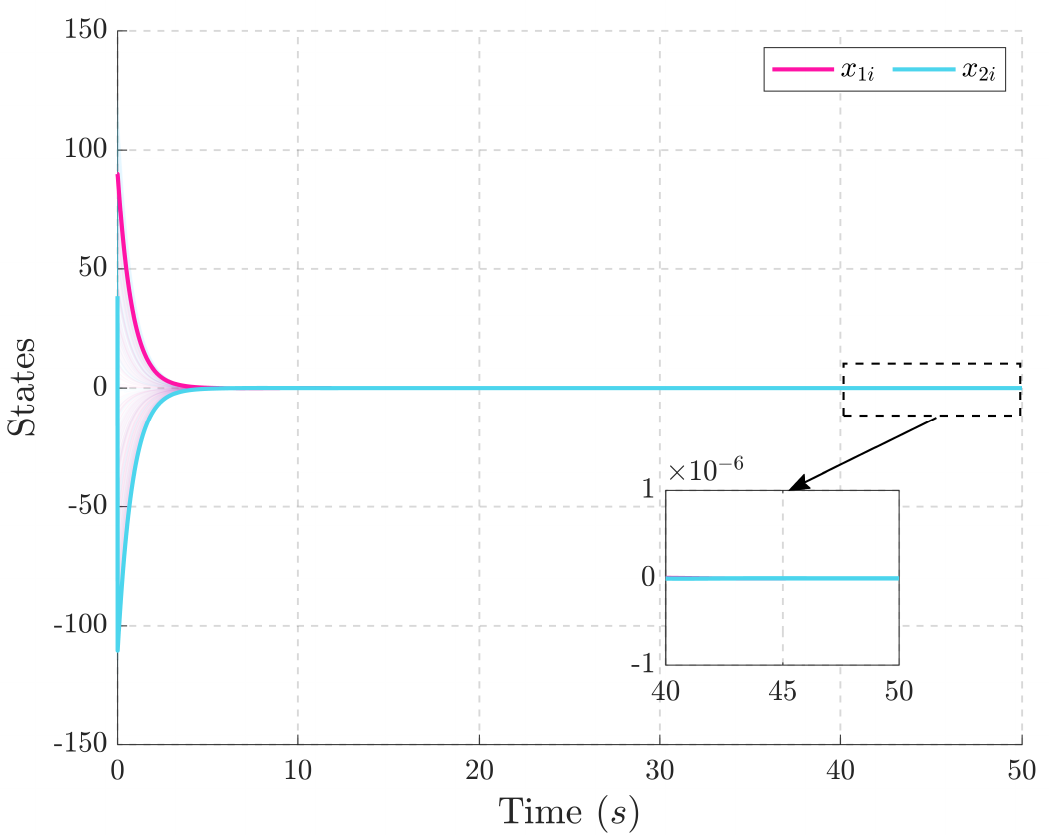}}
	\subfloat[\centering Evolution of closed-loop subsystems with known $D_i$\label{fig:acad known}]{
		\includegraphics[width=0.33\linewidth]{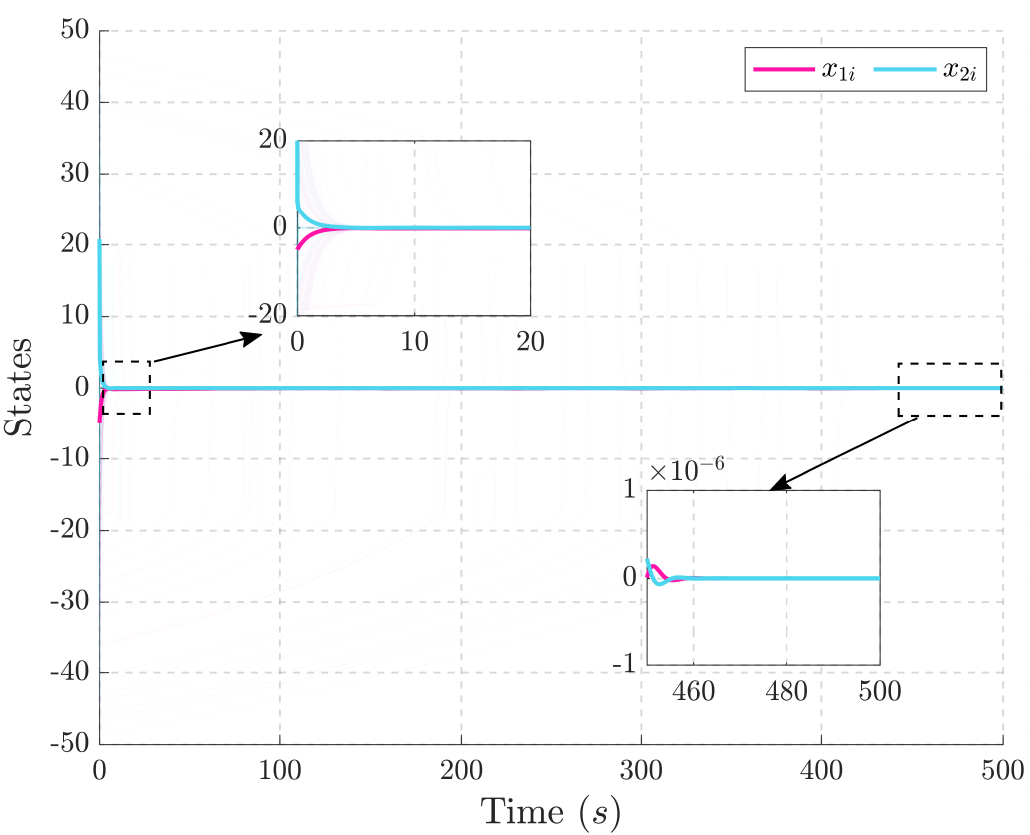}}
	\caption{\textbf{Academic system:}
		Illustration of the subsystems’ evolution under the synthesized controller~\eqref{eq:u Academic unknown} (for unknown $D_i$) or \eqref{eq:u Academic known} (for known $D_i$), where the states of a representative subsystem are highlighted in bold and the remaining subsystems are shown in a faded style. Fig.~(b) corresponds to the case where $D_i$ is unknown and Fig.~(c) corresponds to the case where $D_i$ is known, whereas Fig.~(a) demonstrates the open-loop behavior of the subsystems.}
	\label{fig:acad}
\end{figure*}	

\textbf{Known Upper Bound on $\pmb{\|D_i\|_2}$.} Considering {$\|D_i\|_2\leq 0.04$ and $\mathtt{Card}(\mathbf M_i)= 1$}, where each subsystem is affected by its preceding subsystem, forming a cascade topology, we gather {$T = 25$} samples from each subsystem using a sampling time {$\tau = 0.001$} and {compute $\Lambda_i\Lambda_i^\top = 0.75\times 10^{-4}\mathds I_3$, {given the interval of $\Delta_i$,} as proposed in Remark~\ref{rem:noise}}. Following the steps of Algorithm~\ref{Alg:1} with {$\kappa_i = 0.1$ and $\vartheta_i = 0.8$} yields
\begin{align}
	P_i =& \begin{bmatrix}
		365.8842  &  1.3187  & 78.7867\\
		\star & 330.5682 & 103.2651\\
		\star & \star & 448.7929
	\end{bmatrix}\!\!,\notag\\
	\Vs = &~365.8842\,x_{1_i}^2 + 2.6373\,x_{1_i}x_{2_i} + 157.5735\,x_{1_i}x_{3_i} \notag\\
	&+ 330.5682\,x_{2_i}^2 \!+\! 206.5303\,x_{2_i}x_{3_i} \!+\! 448.7929\,x_{3_i}^2,\notag\\
	u_{1_i} =& ~0.12079\,x_{1_i}^2 - 0.47061\,x_{1_i}x_{2_i} + 0.082367\,x_{1_i}x_{3_i} \notag\\
	&+ 0.16851\,x_{2_i}^2 + 0.88141\,x_{2_i}x_{3_i} - 0.77182\,x_{3_i}^2 \notag\\
	&- 348.8162\,x_{1_i} - 18.0649\,x_{2_i} - 11.1064\,x_{3_i},\notag\\
	u_{2_i} = & ~1.6967\,x_{1_i}^2 + 1.7382\,x_{1_i}x_{2_i} + 3.0714\,x_{1_i}x_{3_i}  \notag\\
	&+ 0.61543\,x_{2_i}^2 + 2.1141\,x_{2_i}x_{3_i} - 0.088252\,x_{3_i}^2 \notag\\
	&+ 0.99924\,x_{1_i} - 319.1249\,x_{2_i} - 37.1072\,x_{3_i},\notag\\
	u_{3_i} = & -0.52565\,x_{1_i}^2 - 3.5244\,x_{1_i}x_{2_i} - 0.23862\,x_{1_i}x_{3_i} \notag\\
	&- 1.8931\,x_{2_i}^2 - 0.89654\,x_{2_i}x_{3_i} + 0.15538\,x_{3_i}^2 \notag\\
	&+ 9.2814\,x_{1_i} - 0.10227\,x_{2_i} - 308.8201\,x_{3_i},\label{eq:u Lorenz unknown}
\end{align}
with $\underline{\alpha}_i = 254.7642,$ $\overline{\alpha}_i = 537.23,$ and $\rho_i = 1.0745$. {According to \eqref{eq:SGC}, these values satisfy condition~\eqref{eq:con spectral} with $r(\Omega) \leq \|\Omega\|_{1,1}\leq 0.0422<1$,} where
	\begin{align*}
		\Omega = 
		\begin{bmatrix}
			0 & 0 & 0 & 0 & 0  & \ldots \\
			0.0422 & 0 & 0 & 0 & 0  & \ldots \\
			0 & 0.0422 & 0 & 0 & 0 & \ldots \\
			0 & 0 & 0.0422 & 0 & 0 &  \ldots \\
			\vdots & \ddots & \ddots & \ddots & \ddots  & \ddots
		\end{bmatrix}\!\!.
	\end{align*}	
	Consequently, the CLF and {UGA-stabilizing} controller for the entire network are constructed by $\V = \sum_{i\in\N^+}\mu_i\Vs$ and $u =(u_i)_{i\in\N^+},$ respectively.

\textbf{Known Matrix $\pmb{D_i}$.} We consider a higher connectivity degree with $\mathtt{Card}(\mathbf M_i) =1000$, where each subsystem is influenced by $1000$ subsequent subsystems, {\ie, $\mathbf M_i=\{i+1,\ldots,i+1000\}$,} when $D_i$ is assumed to be available. Then, with a sampling time {$\tau = 0.001$}, we collect {$T = 80$} samples from each subsystem, which results in {$\Lambda_i\Lambda_i^\top = 2.4\times 10^{-4}\mathds I_3$}, based on Remark~\ref{rem:noise}. We follow the procedure in Algorithm~\ref{Alg:1} with {$\kappa_i = 2$ and $\vartheta_i = 1$,} and obtain
\begin{align}
	P_i =&\begin{bmatrix}
		271.6894  & -3.4691 &  20.2664\\
		\star & 275.7661  & 37.4651\\
		\star  & \star & 289.2950
	\end{bmatrix}\!\!,\notag\\
	\Vs =& ~271.6894\,x_{1_i}^2 - 6.9382\,x_{1_i}x_{2_i} + 40.5328\,x_{1_i}x_{3_i} \notag\\
	&+ 275.7661\,x_{2_i}^2 + 74.9303\,x_{2_i}x_{3_i} + 289.295\,x_{3_i}^2,\notag\\
	u_{1_i} =& ~0.049688\,x_{1_i}^2 - 1.446\,x_{1_i}x_{2_i} + 0.08925\,x_{1_i}x_{3_i} \notag\\
	&- 0.42295\,x_{2_i}^2 - 1.7534\,x_{2_i}x_{3_i} + 2.5993\,x_{3_i}^2 \notag\\
	&- 312.6275\,x_{1_i} + 14.653\,x_{2_i} + 16.2998\,x_{3_i}, \notag\\
	u_{2_i} = & ~1.1647\,x_{1_i}^2 + 1.7924\,x_{1_i}x_{2_i} - 1.0954\,x_{1_i}x_{3_i} \notag\\
	&- 0.0011767\,x_{2_i}^2 + 0.24504\,x_{2_i}x_{3_i} \!+\! 0.14335\,x_{3_i}^2 \notag\\
	&+ 28.7541\,x_{1_i} - 312.7414\,x_{2_i} - 9.6052\,x_{3_i},\notag\\
	u_{3_i} = & -0.075981\,x_{1_i}^2 + 1.9493\,x_{1_i}x_{2_i} - 2.6439\,x_{1_i}x_{3_i} \notag\\
	&+ 0.40238\,x_{2_i}^2 - 0.29296\,x_{2_i}x_{3_i} - 0.19753\,x_{3_i}^2 \notag\\
	&+ 11.3175\,x_{1_i} + 0.41759\,x_{2_i} - 281.0524\,x_{3_i},\label{eq:u Lorenz known}
\end{align}
with $\underline{\alpha}_i = 236.9447,$ $\overline{\alpha}_i = 324.1947,$ and $\rho_i = 0.3242$. {The satisfaction of condition~\eqref{eq:con spectral} using~\eqref{eq:SGC} is verified with $r(\Omega) \leq \|\Omega\|_{1,1}\leq 0.6841<1$, where
	\begin{align*}
		\Omega = 10^{-4}\!\times\!
		\begin{bNiceMatrix}[first-row]
			&\Hdotsfor{3}[line-style=mybrace, shorten=0pt]^{\ensuremath{j\in\{2,\ldots,1001\}}}\\
			0 & 6.8411 & \ldots & 6.8411 & 0  & \ldots \\
			0 & 0 & 6.8411 & \ldots & 6.8411  & \ldots \\
			0 & 0 & 0 & 6.8411 & \ldots & \ldots \\
			0 & 0 & 0& 0 & 6.8411 &  \ldots \\
			\vdots & \ddots & \ddots & \ddots & \ddots  & \ddots
		\end{bNiceMatrix}\!.
	\end{align*}
	Thereby, the CLF and {UGA-stabilizing} controller for the entire network are designed by $\V = \sum_{i\in\N^+}\mu_i\Vs$ and $u =(u_i)_{i\in\N^+},$ respectively.} Fig.~\ref{fig:lor} illustrates the evolution of representative subsystems under the designed controllers~\eqref{eq:u Lorenz unknown} and \eqref{eq:u Lorenz known}, starting from large initial conditions on the order of $10^6$, while demonstrating asymptotic convergence of their states to the origin.

\subsection{Academic System}
\label{subsec:academic}
The final case study is concerned with an academic example endowed with a \emph{state-dependent} control input matrix to validate the practical applicability of our proposed data-driven method. We consider an infinite network of  planar subsystems, where each subsystem is affected by subsystems $j\in\mathbf M_i$.

{Each individual subsystem is described by~\eqref{eq:subsystem reform}, with $x_i=[x_{1_i}~x_{2_i}]^\top$, $u_i=[u_{1_i}~u_{2_i}]^\top$,
			$A_i^* =
			\begin{bmatrix}
				0 & 1 & 0\\
				-1 & -1 & 1
			\end{bmatrix}$, and $B_i^* =[0~1]^\top$. Considering $w_i=(w_{ij})_{j\in\mathbf M_i}=(x_j)_{j\in\mathbf M_i}$ and $D_i=(D_{ij})_{j\in\mathbf M_i}$ according to~\eqref{eq:int constraint} and \eqref{eq:D partition}, respectively, where
			$D_{ij} =  10^{-4}\times\begin{bmatrix}
				0 & 0\\
				3 & 0
			\end{bmatrix}$, the influence of other subsystems is captured by
			$$
			D_i w_i
			=
			10^{-4}\times
			\begin{bmatrix}
				0\\[1mm]
				3\sum_{j\in\mathbf M_i}x_{1_j}
			\end{bmatrix}\!\!.
			$$
			Additionally, $\mathcal F_i^*(x_i) = [x_{1_i}~ x_{2_i}~ x_{1_i}x_{2_i}]^\top\!$ and $\mathcal G_i^*(x_i) =x_{2_i}$, which, altogether with the matrices describing the dynamics, are all unknown.} However, we {know the monomials in $\mathcal F_i^*(x_i)$ and $\mathcal G_i^*(x_i)$ have maximum degrees 2 and 1, respectively. Hence, we construct extended dictionaries as} $\mathcal F_i(x_i) = [x_{1_i}~ x_{2_i}~ x_{1_i}x_{2_i}~ x_{1_i}^2~ x_{2_i}^2]^\top$, including the combination of all monomials up to degree 2, and $\mathcal G_i(x_i) = [1~x_{1_i}~x_{2_i}]^\top$, including the combination of all monomials up to degree 1. Based on the dictionary $\mathcal F_i(x_i),$ we consider the transformation matrix as $\Psi_i(x_i) = [\mathds I_2~[0~x_{1_i}]^\top~\operatorname{diag}(x_{1_i},x_{2_i})]^\top$. Moreover, we assume each component of $\Delta_i$ lies within the interval $[-0.01, 0.01]$.
		
		\textbf{Known Upper Bound on $\pmb{\|D_i\|_2}$.} Let $\|D_i\|_2 \leq 0.06$ for an infinite network with an interconnection topology in which each subsystem is influenced by 5 subsequent subsystems, implying that {$\mathbf M_i=\{i+1,\ldots,i+5\}$ and $\mathtt{Card}(\mathbf M_i)=5$}. We first store $T = 21$ samples from each unknown subsystem with sampling time $\tau = 0.008$, as outlined in Algorithm~\ref{Alg:1}. {Given the interval of $\Delta_i$,} we compute $\Lambda_i\Lambda_i^\top = 0.0042\mathds I_2$ as per Remark~\ref{rem:noise}. Then, we fix $\kappa_i = 0.5$ and $\vartheta_i = 0.5,$ and follow the rest of the steps in Algorithm~\ref{Alg:1}, obtaining
		\begin{align}
			P_i =&10^6\times\begin{bmatrix}
				1.5811  &  0.4921\\
				\star  &  0.3779
			\end{bmatrix}\!\!,\notag\\
			\Vs =& ~1581080.5605\,x_{1_i}^2 + 984282.7422\,x_{1_i}x_{2_i} \notag\\
			&+ 377901.9491\,x_{2_i}^2,\notag\\
			u_i = & -0.063245\,x_{1_i}^2 - 24.7236\,x_{1_i}x_{2_i} - 20.0672\,x_{2_i}^2 \notag\\
			&+ 0.29349\,x_{1_i} + 0.41802\,x_{2_i},\label{eq:u Academic unknown}
		\end{align}
		along with $\underline{\alpha}_i = 2.0224\times 10^5,$ $\overline{\alpha}_i = 1.7567\times 10^6,$ and $\rho_i = 1.2649\times 10^4.$ {We accordingly fulfill condition~\eqref{eq:con spectral} with $r(\Omega) \leq \|\Omega\|_{1,1}\leq 0.6254<1$, obtained via~\eqref{eq:SGC}, where
			\begin{align*}
				\Omega = 
				\begin{bNiceMatrix}[first-row]
					&\Hdotsfor{3}[line-style=mybrace, shorten=0pt]^{\ensuremath{j\in\{2,\ldots,6\}}}\\
					0 & 0.1251 & \ldots & 0.1251 & 0  & \ldots \\
					0 & 0 & 0.1251 & \ldots & 0.1251  & \ldots \\
					0 & 0 & 0 & 0.1251 & \ldots & \ldots \\
					0 & 0 & 0& 0 & 0.1251 &  \ldots \\
					\vdots & \ddots & \ddots & \ddots & \ddots  & \ddots
				\end{bNiceMatrix}\!.
			\end{align*}
			By satisfying compositional condition~\eqref{eq:con spectral}, the CLF and {UGA-stabilizing} controller for the entire network are designed as $\V = \sum_{i\in\N^+}\mu_i\Vs$ and $u =(u_i)_{i\in\N^+},$ respectively.}
		
		\textbf{Known  Matrix $\pmb{D_i}$.} For a denser interconnection, where each subsystem is influenced by $500$ subsequent subsystems, {namely $\mathbf M_i=\{i+1,\ldots,i+500\}$ and $\mathtt{Card}(\mathbf M_i)=500$,} we assume $D_i$ is known. Following the provided steps in Algorithm~\ref{Alg:1}, we gather $T = 50$ samples with sampling time $\tau = 0.01$, resulting in $\Lambda_i\Lambda_i^\top = 0.01\mathds I_2$ as per Remark~\ref{rem:noise}. By setting $\kappa_i = 0.5$ and $\vartheta_i = 0.5$, we design
		\begin{align}
			P_i =&10^6\times\begin{bmatrix}
				2.3486  &  1.2416\\
				\star  &  1.2594
			\end{bmatrix}\!\!,\notag\\
			\Vs =& ~2348564.4205\,x_{1_i}^2 + 2483253.4089\,x_{1_i}x_{2_i} \notag\\
			&+ 1259357.4198\,x_{2_i}^2,\notag\\
			u_i = &-0.42206\,x_{1_i}^2 - 15.3001\,x_{1_i}x_{2_i} - 14.6784\,x_{2_i}^2 \notag\\
			&- 0.66122\,x_{1_i} - 1.496\,x_{2_i},\label{eq:u Academic known}
		\end{align}
		together with $\underline{\alpha}_i = 4.4815\times 10^5,$ $\overline{\alpha}_i = 3.1598\times 10^6,$ and $\rho_i = 284.3797$, {satisfying condition~\eqref{eq:con spectral} with $r(\Omega) \leq \|\Omega\|_{1,1}\leq 0.6346<1$, obtained using~\eqref{eq:SGC}, where
			\begin{align*}
				\Omega = 
				\begin{bNiceMatrix}[first-row]
					&\Hdotsfor{3}[line-style=mybrace, shorten=0pt]^{\ensuremath{j\in\{2,\ldots,501\}}}\\
					0 & 0.0013 & \ldots & 0.0013 & 0  & \ldots \\
					0 & 0 & 0.0013 & \ldots & 0.0013  & \ldots \\
					0 & 0 & 0 & 0.0013 & \ldots & \ldots \\
					0 & 0 & 0& 0 & 0.0013 &  \ldots \\
					\vdots & \ddots & \ddots & \ddots & \ddots  & \ddots
				\end{bNiceMatrix}\!.
			\end{align*}
			Hence, the CLF and UGA-stabilizing controller for the entire network are obtained as $\V = \sum_{i\in\N^+}\mu_i\Vs$ and $u =(u_i)_{i\in\N^+},$ respectively.} Fig.~\ref{fig:acad} depicts the evolution of representative subsystems under the designed controllers~\eqref{eq:u Academic unknown} and \eqref{eq:u Academic known}, showing that their states converge asymptotically to the origin.
		
		\section{Conclusion}\label{sec:conclusion}
		This paper developed a scalable data-driven control framework for infinite networks composed of countably many interconnected nonlinear polynomial subsystems with \emph{unknown} dynamics. By relying solely on a {single noisy} input–state trajectory from each subsystem, we constructed ISS Lyapunov functions and their corresponding input-to-state stabilizing controllers for individual subsystems without requiring explicit model knowledge. Leveraging these local data-driven certificates, a compositional CLF and its associated decentralized UGA-stabilizing controller were synthesized via a small-gain-based interconnection analysis, guaranteeing UGAS of the entire infinite network with unknown dynamics. The proposed methodology explicitly addresses the infinite-dimensional nature of the problem while accommodating unknown subsystem dynamics, thereby overcoming fundamental limitations of existing model-based methods, and approaches restricted to finite networks. The theoretical guarantees were validated through multiple representative case studies, demonstrating the versatility and robustness of the proposed framework. A promising direction for future research is to extend the approach to broader classes of nonlinear dynamics beyond polynomial structures.
\vspace{-0.3cm}
\bibliographystyle{agsm}
\bibliography{biblio}

\begin{authorbio}[Mahdieh]{Mahdieh Zaker} received her B.Sc. from K. N. Toosi University of Technology, Tehran, Iran, in 2019, and her M.Sc. from Amirkabir University of Technology (Tehran Polytechnic), Tehran, Iran, both in Electrical Engineering, control major. She is currently a PhD student in the School of Computing at Newcastle University, UK. She is the Best Repeatability Prize Finalist at the $28^{\text{th}}$ ACM International Conference on Hybrid Systems: Computation and Control (HSCC), 2025. Her research interests are (nonlinear) control and systems theory, data-driven techniques, large-scale systems, and formal methods.\vspace{-0.7cm}
\end{authorbio}

\begin{authorbio}[Andrii]{Andrii Mironchenko} was born in 1986 in Odesa, Ukraine. He received his Ph.D. degree in mathematics from the University of Bremen, Germany (2012), and a habilitation degree from the University of Passau, Germany (2023). He was a Postdoctoral Fellow of the Japan Society for Promotion of Science (2013–2014). Since December 2024, he has been with the Department of Mathematics, University of Bayreuth, Germany.
	
Dr. Mironchenko is the author of the monograph “Input-to-State Stability: Theory and Applications” (Springer, 2023). He is an Associate Editor in Systems \& Control Letters (2023 --) and IEEE Transactions on Automatic Control (2026 --). A. Mironchenko is a co-founder and co-organizer of the biennial Workshop series “Stability and Control of Infinite-Dimensional Systems” (SCINDIS, 2016 --) and ISS Online Seminar (2021 --). He is a recipient of IEEE CSS George S. Axelby Outstanding Paper Award (2023), Outstanding Habilitation Award of the University of Passau (2024), Heisenberg grant (2024), and von Kaven Award (2025) from the German Research Foundation.
	
His research interests include stability theory, nonlinear systems theory, distributed parameter systems, hybrid systems, and applications of control theory to biological systems and distributed control. \vspace{-0.7cm}
\end{authorbio}

\begin{authorbio}[Amy]{Amy Nejati} is an Assistant Professor in the School of Computing at Newcastle University in the United Kingdom. Prior to this, she was a Postdoctoral Associate at the Max Planck Institute for Software Systems in Germany from July 2023 to May 2024. She also served as a Senior Researcher in the Computer Science Department at the Ludwig Maximilian University of Munich (LMU) from November 2022 to June 2023. She received a PhD in Electrical Engineering from the Technical University of Munich (TUM) in 2023. She has received the B.Sc. and M.Sc. degrees both in Electrical Engineering. Her research mainly focuses on developing efficient (data-driven) techniques to design and control highly-reliable autonomous systems while providing mathematical guarantees. She was named a CPS Rising Star 2024, and her paper was selected as a Best Repeatability Prize Finalist at ACM HSCC 2025.\vspace{-0.5cm}
\end{authorbio}
\newpage
\begin{authorbio}[Abolfazl]{Abolfazl Lavaei} is an Assistant Professor in the School of Computing at Newcastle University, United Kingdom. Between January 2021 and July 2022, he was a Postdoctoral Associate in the Institute for Dynamic Systems and Control at ETH Zurich, Switzerland. He was also a Postdoctoral Researcher in the Department of Computer Science at LMU Munich, Germany, between November 2019 and January 2021. He received the Ph.D. degree in Electrical Engineering from the Technical University of Munich (TUM), Germany, in 2019. He obtained the M.Sc. degree in Aerospace Engineering with specialization in Flight Dynamics and Control from the University of Tehran (UT), Iran, in 2014. He is the recipient of several international awards in the acknowledgment of his work including  Best Repeatability Prize (Finalist) at the ACM HSCC 2025, IFAC ADHS 2024, and IFAC ADHS 2021, HSCC Best Demo/Poster Awards 2022 and 2020, IFAC Young Author Award Finalist 2019, and Best Graduate Student Award 2014 at University of Tehran with the full GPA (20/20). His research interests revolve around the intersection of Control Theory, Formal Methods in Computer Science, and Statistical Learning Theory.
\end{authorbio}

\end{document}